\documentclass[lettersize,journal]{IEEEtran}
\usepackage{amsmath,amsfonts}
\usepackage[noend]{algpseudocode}
\usepackage{algorithm}
\usepackage{array}
\usepackage{textcomp}
\usepackage{stfloats}
\usepackage{url}
\usepackage{verbatim}
\usepackage{graphicx}
\usepackage{cite}
\usepackage{multirow}
\usepackage{amssymb,amsthm}
\usepackage{blkarray}
\usepackage{enumerate}

\newtheorem{definition}{Definition}
\newtheorem{example}{Example}
\newtheorem*{example*}{Example}
\newtheorem{lemma}{Lemma}
\newtheorem{theorem}{Theorem}
\newtheorem{remark}{Remark}
\newtheorem{corollary}{Corollary}
\newtheorem{note}{Note}

\usepackage{mathtools}
\newcommand\myeq{\stackrel{\mathclap{\normalfont\mbox{$*$}}}{=}}

\begin{document}

\title{Improved Hotplug Caching Schemes Using PDAs and t-Designs}

\author{Charul Rajput and B. Sundar Rajan \\ Department of Electrical Communication Engineering, Indian Institute of Science, Bengaluru \\
E-mails: charulrajput@iisc.sc.in, bsrajan@iisc.ac.in}



\maketitle

\begin{abstract}
We consider a coded caching system in which some users are offline at the time of delivery. Such systems are called hotplug coded caching systems. A placement delivery array (PDA) is a well-known tool for constructing a coded caching scheme for dedicated caches. In this paper, we introduce the concept of PDAs for hotplug coded caching schemes and refer to it as a hotplug placement delivery array (HpPDA). We give an algorithm to describe the placement and the delivery phase of a hotplug coded caching scheme using HpPDA. We show that an existing hotplug coded caching scheme given by Y. Ma and D. Tuninetti in 2022 corresponds to a class of HpPDAs and then propose a method to further improve the rate of that scheme. Additionally, we construct a class of HpPDAs using $t$-designs, which corresponds to a scheme for hotplug coded caching systems. We further improve the rate of this scheme and prove that the cut-set bound is achieved in some higher memory range for a hotplug coded caching system with three active users.
	\end{abstract}
	

\section{Introduction}
Coded caching is an emerging tool to reduce transmission load during peak traffic hours in the network. The concept of coded caching was first introduced in by Maddah-Ali and Niesen in \cite{MAN2014}. In a general setup of coded caching systems, there is a server with $N$ files and $K$ users connected to the server by an error-free shared link. Each user is equipped with a cache of size equivalent to $M$ files. A coded caching scheme consists of two phases: the placement phase and the delivery phase. In the placement phase, the server stores some content in each user’s cache when the network is not busy and the users have not presented their demands yet. In the delivery phase, when all users have revealed their demands, the transmission from the server will be made in such a way that each user will get their demanded files using transmissions and the cache content. The aim of any coded caching scheme is to jointly design such placement and delivery phases that minimize transmission load at the time of delivery.

Several coded caching schemes have been proposed in the literature \cite{MAN2014, AYG2016, CFL2016, YMA2017, GR2017, MKR2021}. The first and popular scheme presented in \cite{MAN2014} is referred to as the Maddah-Ali and Niesen (MAN) scheme. In \cite{WTP2020}, this scheme was proved to achieve the converse bound on the rate when the placements are uncoded, and the number of files is not less than the number of users. For the case when the number of files is less than the number of users, the authors of \cite{YMA2017} have proposed an improved version of the MAN scheme, referred to as the YMA (Yu, Maddah-Ali and Avestimehr) scheme that achieves the converse bound on the rate given in \cite{YMA2017}. There are many variations of the coded caching model, 
such as coded caching systems with decentralized server \cite{MN2014}, hierarchical coded caching \cite{KNMD2016}, coded caching with nonuniform demands \cite{NM2016}, multi-antenna coded caching \cite{SCK2017}, coded caching for multi-level popularity and access \cite{HKD2017}, coded caching with private demands \cite{WC2020}. 

Another variation of the basic coded caching model is the hotplug coded caching system \cite{MT2022}, in which some users are offline at the time of delivery. Clearly, this is more practical than the original setup because some users may fail to convey their demands to the server; however, in the MAN scheme, the demands of all users are needed to generate the transmissions. In a hotplug coded caching system, there are a total of $K$ users, but only $K'$ users are active at the time of delivery. During the placement, only the number of active users should be known to the server, not their identity. That means placements in the users' cache will be made in such a way that in the delivery phase, the server can send the transmissions as soon as the demands of any $K'$ active users are received. In \cite{MT2022}, authors have provided a scheme for hotplug coded caching systems based on the MAN scheme, in which coded placements were made using Maximum Distance Separable (MDS) codes \cite{LX2004}. An $[n, k]$ MDS code has the property that any $k$ symbols out of $n$ symbols are information symbols, i.e., the information can be recovered from any set of $k$ code symbols. Also, in \cite{LT2021}, authors have considered the case of user inactivity for centralized and decentralized coded caching systems. Recently, in \cite{MT2023}, the hotplug coded caching systems were studied with the demand privacy condition.

The high subpacketization was the problem of the MAN scheme when it came to the implementation. To tackle that problem, in \cite{YCTC2017}, authors proposed the concept of placement delivery array (PDA), which provides both placement and delivery strategies in a single array. They proposed a new centralized coded caching scheme using PDAs and also proved that the MAN scheme corresponds to a class of PDA referred to as MAN PDA. There are many coded caching schemes proposed in the literature based on the PDAs \cite{YTCC2017, SZG2018, CJYT2019, MW2020, ZCJ2020, CJTY2020}.

\subsection{Contributions}

The contributions of this paper are listed as follows.
\begin{enumerate}
\item We introduce the concept of placement delivery arrays for hotplug coded caching schemes and refer to it as HpPDA.
\item Corresponding to each HpPDA, we define the placement and delivery phases of a hotplug coded caching scheme in Algorithm \ref{algo1}.
\item We provide a class of HpPDA that corresponds to the first new scheme given in \cite{MT2022}. Since the placement and delivery phases of the first new scheme were similar to the MAN scheme but with coded placements, the corresponding class of HpPDA consists of MAN PDA. Therefore, we refer to this class as MAN HpPDA.
\item We provide a method to further reduce the rate of a hotplug coded caching scheme using HpPDA, and using that method, we reduce the rate of the scheme that corresponds to MAN HpPDA.
\item We propose a class of HpPDAs based on $t$-designs, and the hotplug coded caching scheme corresponds to that is referred to as $t$-Scheme in this paper. $t$-Scheme provides multiple non-trivial memory rate points using one $t$-design. This is in contrast to the work in \cite{ASK2019, LC2022, CWEC2023, KMR2021} using designs where the schemes corresponds to only one non-trivial point. The rate of $t$-Scheme is further improved.
\item Further, we present the comparison of $t$-Scheme with the existing schemes for several examples of $t$-designs, which show that $t$-Scheme is better than the existing schemes in certain higher memory ranges.
\item Using two classes of $3$-$(v,k,\lambda)$ design, we show that $t$-Scheme for a hotplug coded caching system with three active users achieves the converse lower bound for some memory points.
\end{enumerate}

\subsection{Paper organization}
This paper focuses on the hotplug coded caching systems \cite{MT2022}. In the next section, we describe the model of a hotplug coded caching system and briefly review the scheme of \cite{MT2022}. The definitions of PDA and some combinatorial designs are given in Section \ref{pre} with some results on $t$-designs, which will be useful in the proposed construction of HpPDAs. In Section \ref{HpPDA}, we define HpPDA, and corresponding to each HpPDA, we define a hotplug coded caching scheme whose placement and delivery phases are defined in Algorithm \ref{algo1} given in Section \ref{HpPDA}. Then, we provide a class of HpPDA that corresponds to the scheme given in \cite{MT2022}. We summarize our results in Section \ref{MR}. In Section \ref{improved}, we provide a method to reduce the rate of a hotplug coded caching scheme using HpPDAs. Based on that method, we reduce the rate of the scheme given in \cite{MT2022}. In Section \ref{tdesignPDA}, we construct a class of HpPDA based on a class of combinatorial designs called $t$-designs, and then we improve the rate of that scheme. The comparison between the rate of the proposed scheme with the existing schemes is given in Section \ref{num}. In the last, we prove that the cut-set bound is achieved in some higher memory ranges for a hotplug coded caching system with three active users in Section \ref{OPTIMAL}.
Finally, we conclude the paper in Section \ref{conclu}.

\subsection{Notations}
For a positive integer $n$, $[n]$ denotes the set $\{1,2,\ldots,n\}$ and $[a:b]$ denotes the set $\{a, a+1,\ldots, b\}$. For two sets $A$ and $B$, the notation $A \setminus B$ denotes the set of elements in A which are not in B. For a set $A$, the number of elements in it is represented by $|A|$. The binomial coefficient represented as $\binom{n}{k}$ is equal to $\frac{n!}{k!(n-k)!}$, where $n$ and $k$ are positive integers such that $k \leq n$. For a set $S$ and a positive integer $t$, the notation ${S \choose t}$ represents the collection of all the subsets of $S$ of size $t$. 

\section{Hotplug coded caching system}\label{sec2}
In a $(K,K',N)$ hotplug coded caching system, $K' \leq K$, a server stores $N$ files, denoted by $W_1, W_2, \ldots, W_N$ each of size $B$ bits. The $K$ number of users are connected to the server via an error-free shared link. Each user has a cache of size of $M$ files.

In placement phase, server is unware of which $K'$ users will be active and what will be their demands. Assumption is that server knows that there will be $K'$ number of users active at the time of delivery. In delivery phase, the active users will send their demand to the server, the server will send the transmissions on shared link in such a way that the demand of each active user is statisfied using transmissions and cache content.

In \cite{MT2022}, authors proposed three schemes named as \textit{baseline scheme}, \textit{first new scheme} and \textit{second new scheme}. We briefly describe the \textit{ first new scheme} in the following subsection. 

\subsection{Existing scheme for hotplug caching system based on MAN scheme  \cite{MT2022}}
We will refer to this scheme as MT (Ma and Tuninetti) scheme in the rest of the paper. Fix $t \in [K']$ and partition each file into ${K'\choose t}$ equal-size subfiles as
$$W_i = \{W_{i,S} : S \subset [K'], |S|=t\}, \ \text{for all} \ i \in [N].$$
Then, for every $i \in [N]$, we treat the subfiles of each file as the information symbols of an MDS code with generator matrix $G$ of dimension ${K'\choose t} \times {K\choose t}$. The
MDS coded symbols are
$$ \begin{bmatrix}
C_{i, S'_1} \\
C_{i, S'_2}\\
\vdots \\
C_{i, S'_{{K\choose t}}}
\end{bmatrix} = G^{T}  \begin{bmatrix}
W_{i, S_1} \\
W_{i, S_2}\\
\vdots \\
W_{i, S_{{K'\choose t}}}
\end{bmatrix}, \ \forall i \in [N].$$

\noindent Placement Phase: The cache content of user $j$ is
$$Z_j = \{C_{i, S'} : i \in [N], S' \subseteq [K], |S'|=t, j \in S' \}, \ \forall j \in [K].$$

\noindent Delivery phase: Let $I$ denote the set of active users. Clearly, $I\subseteq [K]$ and $|I|=K'$. The demands of active $K'$ users are denoted by $D[I]=[d_{i_1}, d_{i_2}, \ldots, d_{i_{K'}}]$. Then for all $S \subseteq I, |S|=t+1$ server will broadcast the following subfile
$$X_S = \sum_{k \in S} C_{d_k, S\backslash \{k\}}.$$
In case of repeated demands, there are ${K'-r' \choose t+1}$ out of ${K' \choose t+1}$ redundant transmissions, which need not be sent (according to the YMA delivery \cite{YMA2017}), where $r'=\min(N,K')$.

This scheme achieves the rate memory pair
$$(M_t, R_t)=\left( N \frac{{K-1 \choose t-1}}{{K' \choose t}}, \frac{{K' \choose t+1} - {K'-r' \choose t+1}}{{K' \choose t}} \right), \ \forall t \in [K'],$$
where $r'=\min(N,K')$.

\subsection{Converse bounds}
For a $(K,K',N)$ hotplug coded caching system, we can utilize any converse result from the classical coded caching system with $K'$ users and $N$ files. Clearly, the rate of a hotplug system cannot be better than the rate of a system where the server has prior knowledge of the set of $K'$ active users. The following bound was given in \cite{MAN2014} on the rate of the classical coded caching system and is called cut-set bound.

\begin{lemma}[\cite{MAN2014}]\label{cutset}
For $N$ files and $K'$ users each with cache of size $0\leq M\leq N$, the rate $R$ of a classical coded caching system is bounded by 
$$R \geq \max_{s\in\{1,\ldots,\min\{N,K'\}\}} \left( s-\frac{s}{\lfloor N/s\rfloor}M \right).$$
\end{lemma}

The following converse bound was given by Yu et al. in \cite{YMA2018}.
\begin{lemma}[\cite{YMA2018}]\label{bound2}
For $N$ files and $K'$ users each with cache of size $0\leq M\leq N$, the rate $R$ of a classical coded caching system is lower bounded by 
$$R \geq s-1+\alpha- \frac{s(s-1)-\ell(\ell-1)+2\alpha s}{2(N-\ell+1)}M,$$
for any $s \in [\min{(N, K')}], \alpha \in [0, 1]$, where $\ell \in \{1, \ldots, s\}$ is the minimum value such that
$$\frac{s(s-1)-\ell(\ell-1)}{2}+\alpha s \leq (N-\ell+1)\ell.$$
\end{lemma}

\section{Preliminaries}\label{pre}
In this section, we review the definition of PDAs for dedicated coded caching systems, and PDA construction of MAN coded caching scheme. In Section \ref{tdesignPDA}, we will use combinatorial designs to construct PDAs for hotplug coded caching scheme, for which we review some basic definitions and results related to designs in this section.
 
\begin{definition}[Placement Delivery Array \cite{YCTC2017}]
For positive integers $K, F, Z$ and $S$, an $F \times K$ array $P = (p_{j,k})_{j \in[F], k \in [K]}$, composed of a specific symbol ``$*$" and  non-negative integers from $[1:S]$,
is called a $[K, F, Z, S]$- PDA if it satisfies the following conditions:
\begin{enumerate}[{C}1.]
\item The symbol ``$*$" appears $Z$ times in each column,
\item Each integer occurs at least once in the array,
\item For any two distinct entries $p_{j_1,k_1}$ and $p_{j_2, k_2}$, $p_{j_1, k_1} =
p_{j_2, k_2} = s$ is an integer only if
\begin{enumerate}
\item $j_1 \neq j_2, k_1 \neq k_2$, i.e., they lie in distinct rows and distinct columns; and
\item $p_{j_1, k_2} = p_{j_2, k_1} = *$, i.e., the corresponding $2 \times 2$ sub-array formed by rows $j_1, j_2$ and columns $k_1, k_2$ must be of the following form
$$\begin{bmatrix}
s & * \\
* & s
\end{bmatrix} \quad
\text{or} \quad
\begin{bmatrix}
* & s \\
s & *
\end{bmatrix}.$$
\end{enumerate}
\end{enumerate}

\end{definition}

A $[K, F, Z, S]$-PDA $P$ corresponds to a $(K, N, M)$ coded caching scheme with $K$ users, $N$ files, cache memory of size of $M$ files, subpacketization level $F$, $\frac{M}{N}=\frac{Z}{F}$ and rate $R=\frac{S}{F}$. 

\begin{remark}
A PDA $P$ is called a $g$-$[K, F, Z, S]$ regular PDA if each integer $s\in [S]$ appears exactly $g$ times in $P$.  
\end{remark}

\begin{example}
A $3$-$(6, 4, 2, 4)$ PDA,
\[
B=\begin{bmatrix}
* & 2 & * & 3 & * & 1 \\
1 & * & * & 4 & 2 & *\\
* & 4 & 1 & * & 3 & * \\
3 & * & 2 & * & * & 4
\end{bmatrix}.
\]
\end{example}

\noindent \textbf{Construction of MAN PDA \cite{YCTC2017}:}
Let $K$ be an positive integer, $t \in [K]$ and $F={K \choose t}$. Now arrange all the subsets of size $t+1$ of $[K]$ in a some specific order and define $f(S)$ to be its order  
for a set $S\subseteq [K], |S|=t+1$. Then MAN PDA $P=(p_{\tau,j})$, whose rows are indexed by the elements $\tau$ in ${[K]\choose t}$ and columns are index by the elements $j$ in $[K]$, is defined as
$$p_{\tau, j}= \begin{cases} f(\tau \cup \{j\}) & \text{if} \ j \notin \tau \\
* & \text{otherwise} \end{cases}.$$
Here, $P$ is a $(t+1)$-$\left[K, {K \choose t}, {K-1 \choose t-1}, {K \choose t+1}\right]$  regular PDA, which corresponds to a $(K, N, M)$ coded caching scheme with $K$ users, $N$ files, cache memory of size of $M$ files, $\frac{M}{N}=\frac{t}{K}$ and rate $R=\frac{K-t}{t+1}$.

\begin{example}
A $4$-$[4,6,3,4]$ MAN PDA for $t=2$,
\[
B=\begin{blockarray}{ccccc}
& 1 & 2 & 3 &4   \\
\begin{block}{c[cccc]}
\{1,2\} & * & * & 1 & 2 \\
\{1,3\} & * & 1 & * & 3\\
\{1,4\} & * & 2 & 3 & * \\
\{2,3\} & 1 & * & * & 4\\
\{2,4\} & 2 & * & 4 & * \\
\{3,4\} & 3 & 4 & * & *\\
\end{block}
\end{blockarray}.
\]
\end{example}

Now we present some basic definitions of designs. For more details of these results,  refer to \cite{S2004}.   
\begin{definition}[Design]
A design is a pair $(X, \mathcal{A})$ such that the following properties are satisfied:
\begin{enumerate}
\item $X$ is a set of elements called points, and
\item $\mathcal{A}$ is a collection of nonempty subsets of $X$ called blocks.
\end{enumerate}
\end{definition}

\begin{definition}[$t$-design]
Let $v, k, \lambda$ and $t$ be positive integers such that $v > k \geq t$. A $t$-$
(v, k, \lambda)$ design is a design $(X, \mathcal{A})$ such that the following properties are satisfied:
\begin{enumerate}
\item $|X| = v$,
\item each block contains exactly $k$ points, and
\item every set of $t$ distinct points is contained in exactly $\lambda$ blocks.
\end{enumerate}
\end{definition}
For the sake of  simplified presentation, the parantheses and commas have been dropped from the blocks in the following examples and in some other examples later, i.e., a block $\{a, b, c\}$ is written as $abc$.
\begin{example}
Let $X=\{1,2,3,4,5,6,7\}$, and $\mathcal{A}=\{127,145,136,467,256,357, 234\}$. This is a $2$-$(7,3,1)$ design.
\end{example}

\begin{example}
Let $X=\{1,2,3,4,5,6\}$, and $\mathcal{A}=\{1456, 2356, 1234, 1256, 1346, 2345, 1236, 2456, 1345, 2346,$ $1356, 1245, 3456, 1246, 1235\}$. This is a $3$-$(6,4,3)$ design.
\end{example}

Corresponding to a $t$-$(v,k,\lambda)$ design $(X, \mathcal{A})$ there exists a $(t-i)$-$(v-i,k-i,\lambda)$ design $(X\backslash Z, \{ A \backslash Z : Z \subseteq A \in \mathcal{A}\})$, where $Z \subseteq X, |Z|=i<t$.

\begin{definition}[Incidence Matrix]
Let $(X, \mathcal{A})$ be a design where $X = \{x_1, x_2, \ldots, x_v\}$ and $\mathcal{A} =\{A_1, A_2, \ldots , A_b\}$. The incidence matrix of $(X, \mathcal{A})$ is the $v \times b$ matrix $M = (m_{i, j})$ such that
$m_{i,j} =1 \ \text{if}\ x_i \in A_j$, and $m_{i,j} =0 \ \text{if}\ x_i \notin A_j$.
\end{definition}

\begin{theorem}[\cite{S2004}]\label{thm1.1}
Suppose that $(X, \mathcal{A})$ is a $t$-$(v, k, \lambda)$ design. Suppose that $Y \subseteq X$, where $|Y| = s \leq t$. Then there are exactly
\begin{equation}\label{lambda-s}
\lambda_s =\frac{\lambda{v-s \choose t-s}}{{k-s \choose t-s}}
\end{equation}
blocks in $\mathcal{A}$ that contain all the points in $Y$.
\end{theorem}

Therefore, the number of blocks in a $t$-design is $b=\lambda_0 = \frac{\lambda{v \choose t}}{{k \choose t}}$, and each point belongs to exactly $\lambda_1=\frac{\lambda{v-1 \choose t-1}}{{k-1 \choose t-1}}$ blocks. Clearly, a $t$-$(v, k, \lambda)$ design
 $(X, \mathcal{A})$ is also an $s$-$(v, k, \lambda_s)$ design, for all $1 \leq s \leq t$.

Theorem \ref{thm1.1} can be generalized as follows.

\begin{theorem}[\cite{S2004}] \label{thm1.2}
Suppose that $(X, \mathcal{A})$ is a $t$-$(v, k, \lambda)$ design. Suppose that $Y, Z \subseteq X$, where $Y \cap Z = \emptyset, |Y| = i, |Z| = j$, and $i + j \leq t$. Then there are exactly
$$\lambda_i^{i+j} =\frac{\lambda{v-i-j \choose k-i}}{{v-t \choose k-t}}$$
blocks in $\mathcal{A}$ that contain all the points in $Y$ and none of the points in $Z$.
\end{theorem}

The next corollary directly follows from Theorem \ref{thm1.2}.
\begin{corollary}
Suppose that $(X, \mathcal{A})$ is a $t$-$(v, k, \lambda)$ design and  $ Y \subseteq T \subseteq X$, where $|T| = t, |Y| = i$ with $i \leq t$. Then there are exactly
\begin{equation}\label{lambda-ts}
\lambda_i^{t} =\frac{\lambda{v-t \choose k-i}}{{v-t \choose k-t}}.
\end{equation}
blocks in $\mathcal{A}$ that contain all the points in $Y$ and none of the points in $T \backslash Y$.
\end{corollary}

\section{PDA for hotplug coded caching system } \label{HpPDA}

In this section, we introduce a PDA structure for hotplug coded caching system, and then provide an algorithm to  characterize the placement and delivery of a hotplug coded caching scheme.

\begin{definition}[Hotplug placement delivery array (HpPDA)]
Let $K, K', F, F`, Z, Z'$ and $S$ be integers such that $K \geq K', F \geq F'$ and $Z<F'$. Consider two arrays given as follows
\begin{itemize}
\item $P=(p_{f,k})_{f \in [F], k \in [K]}$ which is an array containing `$*$' and null. Each column contains $Z$ number of `$*$'s.
\item $B=(b_{f,k})_{f \in [F'], k \in [K']}$ which is a $[K', F', Z', S]$-PDA.

\end{itemize}

For each $\tau \subseteq [K], |\tau| = K'$, there exists a subset $\zeta \subseteq [F], |\zeta|=F'$ such that 
\begin{equation}\label{c}
[P]_{\zeta \times \tau} \myeq B,
\end{equation}
where $[P]_{\zeta \times \tau}$ denotes the subarray of $P$ whose rows corresponds to the set $\zeta$ and columns corresponds to the set $\tau$, and $\myeq$ denotes that the positions of all `$*$'s are same in both the arrays. We call it a $(K,K',F, F',Z,Z',S)$- HpPDA $(P, B)$.

\end{definition}
\begin{example}\label{ex1}
For the parameters $K=6, K'=3, F=6, F'=3, Z=Z'=1$ and $S=3$, a HpPDA $(P, B)$ is given by
\begin{align*}
B&=\begin{bmatrix}
* & 1 & 2 \\
 1 & * & 3\\
 2 & 3 & * 
\end{bmatrix}, \ 
P=\begin{bmatrix}
* & & & & &   \\
 & * & & & & \\
 & & * & & & \\
 & & & * & & \\
 & & & & * & \\
 & & & & & *
\end{bmatrix}.
\end{align*}
Clearly, for each $\tau \subseteq [6], |\tau|=3$, there exists $\zeta \subseteq [6], |\zeta|=3$ such that $[P]_{\zeta \times \tau} \myeq B$. In this example, $\zeta=\tau$ for all $\tau \subseteq [6]$.

\end{example}

\begin{example}\label{ex2}
For the parameters $K=6, K'=4, F=15, F'=6, Z=5, Z'=3, S=4$, a HpPDA $(P, B)$ is given as
\[
B=\begin{bmatrix}
* & * & 1 & 2 \\
* & 1 & * & 3\\
* & 2 & 3 & * \\
1 & * & * & 4\\
2 & * & 4 & * \\
3 & 4 & * & *
\end{bmatrix}, \ 
P=\begin{blockarray}{ccccccc}
& 1 & 2 & 3 &4 & 5& 6  \\
\begin{block}{c[cccccc]}
1& * &* & & & &   \\
2 & * &  & * & & & \\
3 & * & &   &* & & \\
4 & * & & & &* & \\
5 & * & & & &  &* \\
6 & &* & *& & & \\
7 & & * & & * & &\\
8 & & * & & & * & \\
9 & & * & & & & * \\
10 & & & * & * & & \\
11 & & & * & & * & \\
12 & & & * & & & * \\
13 & & & & * & * & \\
14 & & & & * & & * \\
15 & & & & & * & *\\
\end{block}
\end{blockarray}.
\]
Here, for each $\tau \subseteq [6], |\tau|=4$, there exists $\zeta =\{S \subseteq \tau \ |\ |S|=2\}, |\zeta|=6$ such that $[P]_{\zeta \times \tau} \myeq B$. 
\end{example}

\begin{theorem}\label{thm2.1}
Given a $(K,K',F, F',Z,Z',S)$-HpPDA $(P, B)$, there exists a $(K, K', N)$ hotplug coded caching scheme with the following parameters,
\begin{itemize}
\item subpacketization level is $F'$,
\item $\frac{M}{N}=\frac{Z}{F'},$ where $M$ denotes the number of files each user can store in it's cache and $M<N$,
\item rate, $R=\frac{S}{F'}$,
\end{itemize}
which can be obtained by Algorithm \ref{algo1}, in which $[F, F']$ MDS code is used for encoding the subfiles of each file,.
\end{theorem}

\begin{proof}
In the placement phase of Algorithm \ref{algo1}, array $P$ is used to fill the cache of each user. The columns of $P$ corresponds to $K$ users and rows corresponds to $F$ coded subfiles $C_{n,f}$ of each file $W_{n}, n \in [N]$. The cache of the user $k$ contains the subfile $C_{n,f}$ if $p_{f,k}=*$. Since each column of $P$ constains $Z$ number of *'s and the size of each coded subfile is $B/F'$, we have $\frac{M}{N}=\frac{Z}{F'}$. Since $Z<F'$, we have $M<N$.

In the delivery phase of Algorithm \ref{algo1}, array $B$ is used for the transmissions. Let $I=\{i_1, i_2, \ldots, i_{K'}\},$ $1 \leq i_k \leq K, \forall k \in [K']$, be the set of active users with demands $D[I]=(d_{i_1}, d_{i_2}, \ldots, d_{i_{K'}})$. By the property of HpPDA, for set $I$ there exists a subset $\zeta \subseteq [F], |\zeta|=F'$ such that $[P]_{\zeta, I} \myeq B$.

Then a PDA $\overline{P}$ is constructed for users with the help of $B$, and corresponding to each integer $s \in [S]$ in $\overline{P}$, a $X_s$ transmission is made. Since $\overline{P}$ is a PDA, each integer in $i_k$-th column will provide user $i_k$ a coded subfile of the desired file $d_{i_k}$. There are total $Z'$ stars and $F' - Z'$ integers in a column. Hence using the transmissions, each user will get $F'-Z'$ coded subfiles of the desired file, and there are already $Z$ coded subfiles of every file in each user’s cache. Therefore, the user have total $F'-Z'+Z \geq F'$ (as $Z-Z' \geq 0$) number of coded subfiles of the desired file, and using those subfiles, the user can recover the desired file.
Since the total number of transmissions is equal to the total number of integers in PDA $B$, which is equal to $S$, therefore, the rate is $R=\frac{S}{F'}$.

\end{proof}

\begin{algorithm}
	\caption{Hotplug coded caching scheme based on HpPDA $(P, B)$ }
	
	\begin{algorithmic}[1]
		
\Function{Placement phase}{$P; W_n, n \in[N]; G$}
\State Divide each file $W_n, n \in [N]$ into $F'$ subfiles, i.e., 
$W_n=\{W_{n, f'} | f' \in [F']\}.$

\State Encode $F'$ subfiles of each file $W_n, n \in [N]$, into $F$ coded subfiles using $[F, F']$ MDS code with generator matrix $G$ of order $F' \times F$, i.e.,
$$\begin{bmatrix}
C_{n,1} \\ C_{n,2} \\ \vdots \\ C_{n, F} 
\end{bmatrix} =  G^{T}\begin{bmatrix}
W_{n,1} \\ W_{n,2} \\ \vdots \\ W_{n, F'} 
\end{bmatrix} .$$
\For{$k \in [K]$}
\State $Z_k \gets \{C_{n,f} \ | \ p_{f,k}=*, n \in [N], f \in [F]\}$.
\EndFor
\EndFunction 

\Function{Delivery phase}{$B; P; W_n, n \in[N]; I; D[I]$}
\State Let $I$ be the set of active users with demands $D[I]=(d_{i_1}, d_{i_2}, \ldots, d_{i_{K'}})$. 
\State By the property of $(P, B)$ HpPDA, for $I \subseteq [K], |I| = K'$, there exists a subset $\zeta \subseteq [F], |\zeta|=F'$ such that 
$[P]_{\zeta \times I} \myeq B.$
\State Make a new array $\overline{P}=(\overline{p}_{f,k})_{f \in \zeta, k \in I}$ by filling $s \in [S]$ integers in null spaces of the subarray $[P]_{\zeta \times I}$ in such a way that $\overline{P}=B$.
 \For{$s \in [S]$}
\State Server sends the following coded subfiles:
\State $$X_s = \sum_{\overline{p}_{f,k} = s, f \in \zeta, k \in I } C_{d_k, f}.$$
\EndFor
\EndFunction 
		
	\end{algorithmic}
	\label{algo1}
\end{algorithm}

\subsection{HpPDA for the MT scheme}\label{subsec}
In this section, we present a class of HpPDA for a $(K, K', N)$ hotplug coded caching scheme given in Section \ref{sec2}. We refer to HpPDAs in this class as MAN HpPDAs. For $t \in [K']$, let
\begin{align}
F'&={K' \choose t}, \qquad \ \ F={K \choose t}, \nonumber \\ 
Z'&={K'-1 \choose t-1}, \quad  Z={K-1 \choose t-1}, \label{para} \\
S&={K' \choose t+1}.  \nonumber 
\end{align}
Further, let $B$ be an $[K', F', Z', S]$-PDA for MAN scheme for $K'$ users, and $P$ can be obtained by replacing all the integers by null in an $[K, F, Z, S_1]$-PDA for MAN scheme for $K$ users. Therefore, we have $B=(b_{\tau, k})_{\tau \in {[K'] \choose {t}}, k \in [K']}, \ \text{where} $
\begin{equation}\label{B}
 b_{\tau, k}=\begin{cases}
* & \text{if} \ k \in \tau \\
f(\tau \cup \{k\}) & \text{if} \ k \notin \tau
\end{cases},
\end{equation}
and $f$ is an one to one map from ${[K'] \choose {t+1}}$ to $[{K' \choose {t+1}}]$, and the array $P=(p_{\tau, k})_{\tau \in {[K] \choose {t}}, k \in [K]}, \ \text{where} $
\begin{equation}\label{P}
 p_{\tau, k}=\begin{cases}
* & \text{if} \ k \in \tau \\
\text{null} & \text{if} \ k \notin \tau
\end{cases}.
\end{equation}

\begin{note}
The HpPDA given in this subsection corresponds to the MT scheme without YMA delivery, i.e., the rate is $R=\frac{S}{F'}=\frac{{K' \choose t+1}}{{K' \choose t}}$. However, if $N<K'$, some redundant transmissions can be reduced using YMA delivery.
\end{note}

\begin{theorem}
For the parameters defined in \eqref{para}, $(P, B)$ is a $(K,K',F, F',Z,Z',S)$-HpPDA.
\end{theorem}

\begin{proof}
Clearly, $B$ is an $[K', F'={K' \choose t}, Z'={K'-1 \choose t-1}, S={K' \choose t+1}]$ PDA obtained from MAN scheme for $K'$ users, and $P$ is a ${K \choose t} \times K$ array whose rows are indexed by the elements in the set ${[K]\choose t}$ and columns are indexed by the elements in $[K]$. The array $P$ contains exactly ${K-1 \choose t-1}$ ``$*$"\textquotesingle s in each column. Now we need to prove that the pair $(P, B)$ satisfy the property \eqref{c}.

Let $\tau=\{i_1,i_2, \ldots, i_{K'}\} \subseteq [K]$. Consider a set $\zeta = \{ S \subseteq {[K] \choose t} \mid S \subseteq \tau\}=\{ S \subseteq \tau \mid |S|=t\}={\tau \choose t}$. Clearly, $|\zeta|={K' \choose t}=F'$. Corresponding to the sets $\tau$ and $\zeta$, we have a subarray $[P]_{ \zeta\times \tau}$ of $P$ such that $[P]_{\zeta\times \tau}(S, i_j)=*$ if and only if $i_j \in S$, $S\in \zeta$. Also, $B(S,j)=*$ if and only if $j \in S$, $S \in {[K'] \choose t}$. There is an one to one correspondence $\phi : \tau \to [K']$ such that $\phi(i_j)=j, \forall i_j \in \tau$. Also, there is an one to one correspondence $\psi: \zeta={\tau \choose t} \to {[K'] \choose t}$ such that $\psi(\{i_{j_1}, i_{j_2}, \ldots, i_{j_t}\})=\{j_1, j_2, \ldots, j_t\}$, for all $1 \leq j_1, j_2, \ldots, j_t \leq K'$. Clearly, $i_j \in S$, $S\in \zeta$, if and only if $\phi(i_j) \in \psi(S)$. Therefore, we have $[P]_{\zeta\times \tau}(S, i_j)=B(\psi(S),\phi(i_j))=*$ if and only if $i_j \in S$, $S\in \zeta$. Hence $[P]_{\zeta\times \tau} \myeq B$.
\end{proof}

\begin{example*}[Example \ref{ex2} continued]
Consider the HpPDA $(P, B)$ given in Example \ref{ex2} which corresponds to a hotplug coded caching system with $K=6$ users and $K'=4$ active users. This HpPDA can be achieved by the method given in Subsection \ref{subsec} for $t=2$.

Now using HpPDA $(P, B)$ in Algorithm \ref{algo1}, we get a $(6,4,6)$ hotplug coded caching scheme in the following way.

\noindent \textbf{Placement Phase:} From Line 2 in Algorithm \ref{algo1}, each file is divided into $F'=6$ subfiles, i.e., $W_n=\{W_{n, 1}, W_{n, 2}, \ldots, W_{n,6}\}, n \in [6]$. Encode all the subfiles of file $W_n$ into $F=15$ coded subfile $C_{n,1}, C_{n,2}, \ldots, C_{n,15}$, using a $[15,6]$ MDS code with generator matrix $G_{6\times 15}$, i.e.,
$$\begin{bmatrix}
C_{n,1} & C_{n,2} & \cdots & C_{n, 15} 
\end{bmatrix} =  \begin{bmatrix}
W_{n,1} & W_{n,2} & \cdots & W_{n, 6} 
\end{bmatrix}G.$$
From Lines 4-6 in Algorithm \ref{algo1}, the caches of all users are filled as follows:
\begin{align*}
Z_1&=\{C_{n,1}, C_{n,2}, C_{n,3}, C_{n,4}, C_{n,5} \ | \ n \in [6]\} \\
Z_2&=\{C_{n,1}, C_{n,6}, C_{n,7}, C_{n,8}, C_{n,9} \ | \ n \in [6]\} \\
Z_3&=\{C_{n,2}, C_{n,6}, C_{n,10}, C_{n,11}, C_{n,12} \ | \ n \in [6]\} \\
Z_4&=\{C_{n,3}, C_{n,7}, C_{n,10}, C_{n,13}, C_{n,14} \ | \ n \in [6]\} \\
Z_5&=\{C_{n,4}, C_{n,8}, C_{n,11}, C_{n,13}, C_{n,15} \ | \ n \in [6]\} \\
Z_6&=\{C_{n,5}, C_{n,9}, C_{n,12}, C_{n,14}, C_{n,15} \ | \ n \in [6]\}. 
\end{align*} 
 
\noindent \textbf{Delivery Phase:} Let the set of active users be $I=\{1,4,5,6\}$ with demands $D[I]=\{2,3,1,5\}$. From Line 11 in Algorithm \ref{algo1}, we get
\begin{align*}
&\begin{matrix}
& 1 & 4 & 5 & 6 
\end{matrix}\\
\overline{P}=  \begin{matrix}
3 \\ 4 \\ 5 \\ 13 \\ 14 \\ 15 
\end{matrix} & \begin{bmatrix}
* & * & 1 & 2 \\
* & 1 & * & 3\\
* & 2 & 3 & * \\
1 & * & * & 4\\
2 & * & 4 & * \\
3 & 4 & * & *
\end{bmatrix}.
\end{align*}

From Lines 12-15 in Algorithm \ref{algo1}, the server trasmits the following coded files:
\begin{align*}
X_1&= C_{2,13} + C_{3, 4} +C_{1, 3}    \\
X_2 &= C_{2,14} + C_{3, 5} +C_{5, 3}      \\
X_3 &=   C_{2,15} + C_{1, 5} +C_{5, 4}     \\
X_4 &= C_{3,15} + C_{1, 14} +C_{5, 13}  
\end{align*}
To recover a file, only six coded subfiles are required because the MDS codes used in this example has dimension $6$. The demand of user $1$ is $W_2$. The user $1$ gets $C_{2,13}, C_{2,14}, C_{2,15}$ from $X_1, X_2, X_3$, respectively. Since user $1$ already has $C_{2,3}, C_{2,4}, C_{2,5}$ in its cache, it can recover file $W_2$. Similarly, user $4,5$ and $6$ will get their desired files.
We have
$$\frac{M}{N}=\frac{Z}{F'}=\frac{5}{6} \ \text{and} \ R=\frac{S}{F'}=\frac{2}{3}.$$

\end{example*}

There exist HpPDAs other than MAN HpPDAs given in subsection \ref{subsec} as shown in the following example.
\begin{example}\label{ex3}
Consider a hotplug coded caching system with $K=6$ users and $K'=5$ active users. Then for $F=12, F'=5, Z=4, Z'=2, S=9$, we have HpPDA $(P, B)$, where
\[
B=\begin{bmatrix}
* & 1 & 4 & 6 & * \\
* & * & 2 & 5 & 7\\
1 & * & * & 3 & 8\\
4 & 2 & * & * & 9\\
6 & 5 & 3 & * & *
\end{bmatrix}
P=\begin{blockarray}{ccccccc}
& 1 & 2 & 3 &4 & 5& 6  \\
\begin{block}{c[cccccc]}
1 & * &  & & & * & \\
2 & * & & & & & *  \\
3 & *&  &* & & &  \\
4 &* & * &  & & & \\
5 & & *&  & * & & \\
6 & & * & &  & & *\\
7 & &* & *  & & & \\
8 & & & * & & * & \\
9 & & &* &* & & \\
10 & & & &* & * & \\
11 & & & & & * & *\\
12 & & &  & *& & * \\
\end{block}
\end{blockarray}.
\]
For each $\tau \subseteq [6], |\tau|=5$, there exists $\zeta \subseteq [12], |\zeta|=5$ such that $[P]_{\zeta \times \tau} \myeq B$. A choice of $\zeta$ for each set $\tau$ of size $5$ is given in Table \ref{tab1}.

\begin{table}[!t]
\caption{Example \ref{ex3}: A choice of $\zeta$ for each set $\tau$ of size $5$. \label{tab1}}
\centering
\begin{tabular}{|c|c|} 
 \hline
 $\tau$ & $\zeta$ \\ 
 \hline
 $\{1,2,3,4,5\}$ & $\{1,4,7,9,10\}$ \\ 
 $\{1,2,3,4,6\}$ & $\{2,4,7,9,12\}$ \\ 
 $\{1,2,3,5,6\}$ & $\{2,4,7,8,11\}$ \\ 
 $\{1,2,4,5,6\}$ & $\{2,4,5,10,11\}$ \\ 
 $\{1,3,4,5,6\}$ & $\{2,3,9,10,11\}$ \\ 
$\{2,3,4,5,6\}$ & $\{6,7,9,10,11\}$ \\ 
\hline
\end{tabular}
\end{table}

Now using HpPDA $(P, B)$ in Algorithm \ref{algo1}, we get an $(6,5,6)$ hotplug coded caching scheme in the following way.

\noindent \textbf{Placement Phase:} From Line 2 in Algorithm \ref{algo1}, each file is divided into $F'=6$ subfiles, i.e., $W_n=\{W_{n, 1}, W_{n, 2}, \ldots, W_{n,5}\}, n \in [6]$. Encode all the subfiles of file $W_n$ into $F=12$ coded subfile $C_{n,1}, C_{n,2}, \ldots, C_{n,12}$, using an $[12,5]$ MDS code with generator matrix $G_{5\times 12}$, i.e.,
$$\begin{bmatrix}
C_{n,1} & C_{n,2} & \cdots & C_{n, 12} 
\end{bmatrix} =  \begin{bmatrix}
W_{n,1} & W_{n,2} & \cdots & W_{n, 5} 
\end{bmatrix}G.$$
From Lines 4-6 in Algorithm \ref{algo1}, the caches of all users are filled as follows:
\begin{align*}
Z_1&=\{C_{n,1}, C_{n,2}, C_{n,3}, C_{n,4} \ | \ n \in [6]\} \\
Z_2&=\{C_{n,4}, C_{n,5}, C_{n,6}, C_{n,7} \ | \ n \in [6]\} \\
Z_3&=\{C_{n,3}, C_{n,7}, C_{n,8}, C_{n,9} \ | \ n \in [6]\} \\
Z_4&=\{C_{n,5}, C_{n,9}, C_{n,10}, C_{n,12} \ | \ n \in [6]\} \\
Z_5&=\{C_{n,1}, C_{n,8}, C_{n,10}, C_{n,11} \ | \ n \in [6]\} \\
Z_6&=\{C_{n,2}, C_{n,6}, C_{n,11}, C_{n,12} \ | \ n \in [6]\}. 
\end{align*} 
\noindent \textbf{Delivery Phase:} Let the set of active users be $I=\{1,2,4,5,6\}$ with demands $D[I]=\{6,3,1,2,5\}$. From Line 11 in Algorithm \ref{algo1}, we get

\begin{align*}
&\begin{matrix}
& 1 & 2 & 4 & 5 & 6 
\end{matrix}\\
\overline{P}=  \begin{matrix}
2 \\ 4 \\ 5 \\ 10 \\ 11  
\end{matrix} & \begin{bmatrix}
* & 1 & 4 & 6 & * \\
* & * & 2 & 5 & 7\\
1 & * & * & 3 & 8 \\
4 & 2 & * & * & 9\\
6 & 5 & 3 & * & * 
\end{bmatrix}.
\end{align*}

From Lines 12-15 in Algorithm \ref{algo1}, the server trasmits the following coded files:
$
X_1= C_{6,5} + C_{3, 2}, X_2 = C_{3,10} + C_{1, 4}, X_3 = C_{1,11} + C_{2, 5}, X_4 = C_{6,10} + C_{1, 2}, X_5=C_{3,11}+C_{2,4}, X_6=  C_{6,11}+C_{2,2}, X_7=C_{5,4}, X_8=C_{5, 5}$ and  $X_9=C_{5,10}$.
To recover a file, only five coded subfiles are required because the MDS codes used in this example has dimension $5$. The demand of the user $4$ is $W_1$. The user $4$ gets $C_{1,4}, C_{1,11}, C_{1,2}$ from $X_2, X_3, X_4$, respectively. Since the user $4$ already has $C_{1,5}, C_{1,10}$ in its cache, it can recover file $W_2$. Similarly, the users $1,2,5$ and $6$ will get their desired files. Here, we have
$\frac{M}{N}=\frac{4}{5} \ \text{and} \ R=\frac{9}{5}.$
\end{example}

\section{Main results} \label{MR}
In this section we summarise our main results, which will be proved in the subsequent sections.
In the following theorem, we give the memory-rate points achievable by the improved version of MAN HpPDA given in Section \ref{improved}.

\begin{theorem}\label{MR1}
For $(K, K', N)$ hotplug coded caching system, the lower convex envelope of the following points is achievable
$$\left( \frac{M}{N}, R \right) = \left( \frac{{K-1 \choose t-1}}{{K' \choose t}}, \frac{{K' \choose t+1}-\left \lfloor \frac{K'}{t+1} \right \rfloor ({K-1 \choose t-1}-{K'-1 \choose t-1})}{{K' \choose t}} \right), $$
for all $t \in [0:K']$.
\end{theorem}
In the following theorem, we give the memory-rate points achievable by HpPDAs constructed using $t$-designs in Section \ref{tdesignPDA}. For a $t$-$(v,k,\lambda)$ design, $\lambda_1$ is defined in \eqref{lambda-s} and $\lambda_s^t$ is defined in \eqref{lambda-ts}, for $1\leq s < t$. 
\begin{theorem}\label{MR2}[$t$-Scheme]
For $s \in [t-1], a_s\in \{0,1, \ldots, \lambda_s^t\}$, such that $\sum_{s=1}^{t-1} a_s {t \choose s} > \lambda_1$, 
a $(v, t, N)$ hotplug coded caching scheme exist with the following memory-rate points
$$\left( \frac{M}{N}, R \right) = \left( \frac{\lambda_1}{\sum_{s=1}^{t-1} a_s {t \choose s}}, \frac{\sum_{s=1}^{t-1} a_s {t \choose s+1}}{\sum_{s=1}^{t-1} a_s {t \choose s}} \right).$$
\end{theorem}
The rate given in Theorem \ref{MR2} can be further improved which is described in Subsection \ref{improvescheme2}.
Further, in Section \ref{num}, we give a comparison of improved $t$-Scheme with other existing scheme for multiple examples of $t$-designs and show that we are getting better rate in some higher memory ranges. Also, we get some memory-rate points which achieves the cut-set bound. 

The following two clasess of $t$-designs are given in \cite[Theorem 9.14]{S2004} and \cite[Theorem 9.27]{S2004}, respectively. 
\begin{itemize}
\item For all even integers $v \geq 6$, there exists a $3$-$(v,4,3)$ design. 
\item For all prime powers $q$, there exists a $3$-$(q^2+1, q+1, 1)$ design.
\end{itemize}
Using these clasess of designs, we get some optimal memory-rate points achieving cut-set bound for a $(K, 3, N)$ hotplug coded caching system.
 
\begin{theorem}\label{MR3}
For a $(K, 3, N)$ hotplug coded caching system, where $K$ is an even integer and $K \geq 6$, we get the optimal memory-rate points 
$$\left(\frac{M}{N}, R \right)=\left(\frac{{K-1 \choose 2}}{3(a_1+a_2)}, \frac{3(a_1+a_2)-{K-1 \choose 2}}{3(a_1+a_2)}\right),$$
where $0\leq a_1 \leq {K-4 \choose 2}$, $0 \leq a_2 \leq \frac{3}{2}(K-4)$ and 
${K-1 \choose 2}=a+3a_1+2a_2$ for some $0 \leq a < a_2$. 
\end{theorem}

\begin{theorem}\label{MR4}
For a $(q^2+1, 3, N)$ hotplug coded caching system, where $q$ is a prime power and $q \geq 3$, we get the optimal memory-rate points 
$$\left(\frac{M}{N}, R \right)=\left(\frac{q(q+1)}{3(a_1+a_2)}, \frac{3(a_1+a_2)-q(q+1)}{3(a_1+a_2)}\right),$$
where $0\leq a_1 \leq q^2-q-1$, $0 \leq a_2 \leq q$ and $q(q+1)=a+3a_1+2a_2$ for some $0 \leq a < a_2$. 
\end{theorem}

The proofs of Theorem \ref{MR3} and Theorem \ref{MR4} are given in Section \ref{OPTIMAL}.

\section{Hotplug coded caching scheme with improved rate}\label{improved}

The rate of the scheme given in Theorem \ref{thm2.1} can be improved further if $Z'<Z$. All the coded subfiles $C_{n,1}, C_{n, 2}, \ldots, C_{n, F}$ of a file $W_n, n \in [N]$ are obtained using $[F,F']$ MDS code. Therefore, to recover a file $W_n$, only $F'$
coded files are needed (by the property of MSD codes).

In the placement phase of Algorithm \ref{algo1}, $Z$ number of coded subfiles of each file are placed in each user's cache. In the delivery phase of Algorithm \ref{algo1}, each active user is getting $(F'-Z')$ coded subfiles of desired file. But only $(F'-Z)$ coded subfiles of the desired file are needed for the recovery as $Z$ coded subfiles of every file are already placed in user's cache. That means each user is getting $(Z-Z')$ more coded subfiles of the desired file than required. Since the $j$-th active user gets one coded subfile corresponding to each integer in $j$-th column of array $B$, we can remove $(Z-Z')$ integers from each column of array $B$.

Let $C_i$ denote the set of integers appearing in the $i$-th column of array $B$, where $i \in [K']$. Clearly, $C_i \subseteq [S]$ and $|C_i|=F'-Z'$ for all $i \in [K']$. Choose a subset $T \subseteq [S]$ such that 
\begin{equation}\label{T}
|T \cap C_i| \leq Z-Z', \ \forall i \in [K'].
\end{equation}
Make a new array $\overline{B}$ by replacing all $t \in T$ by null in array $B$. So the array $\overline{B}$ has only $S-|T|$ number of integers. By running Algorithm \ref{algo1} for HpPDA $(\overline{B},P)$, we get a hotplug coded caching scheme with rate,
$$R=\frac{S-|T|}{F'}.$$
Since this can be done for any set $T\subseteq [S]$ which statisfy the condition in \eqref{T}, choose the largest possible set $T$, so that the reduction in rate will be large.

Clearly, If $Z=Z'$, then there will be no improvement in the rate.

\subsection{Improved version of the MT scheme}
Consider the MAN HpPDA $(P, B)$ for the MT scheme, where $B$ and $P$ are defined in (\ref{B}) and (\ref{P}). There are total $Z'={K'-1 \choose t-1}$ number of $*$'s and $F'-Z'={K'-1 \choose t}$ number of integers in each column of array $B$, and an integer $s \in [S]$ appears exactly in $t+1$ number of columns. Therefore, $\left \lfloor \frac{K'}{t+1} \right \rfloor (Z-Z')$ number of integers can be removed from the array $B$ in such a way that no column will have less than $F'-Z$ number of integer, i.e., there exists a set $T \subseteq [S]$ which satisfies (\ref{T}) and $|T|=\left \lfloor \frac{K'}{t+1} \right \rfloor (Z-Z')$. Hence we have
$$R=\frac{S-\left \lfloor \frac{K'}{t+1} \right \rfloor (Z-Z')}{F'}.$$  

\begin{remark}
\begin{itemize}
\item In Example \ref{ex2}, the set $T=\{1,2\}$ satisfies the condition given in (\ref{T}), and therefore we have
$$\overline{B}=\begin{bmatrix}
* & * &  &  \\
* &  & * & 3\\
* &  & 3 & * \\
 & * & * & 4\\
 & * & 4 & * \\
3 & 4 & * & *
\end{bmatrix}.$$
Now by using Algorithm \ref{algo1} for HpPDA $(\overline{B}, P)$, we get rate $R=\frac{1}{3}$.
That means in the delivery phase, the server only transmits $X_3$ and $X_4$. The user $1$ gets coded subfile $C_{2,15}$ using $X_3$. Then user $1$ gets its desired file $W_2$ using $C_{2,1}, C_{2,2}, C_{2,3}, C_{2,4}, C_{2,5}$ from its cache and $C_{2,15}$.  The comparison with the existing schemes \cite{MT2022} is given in Fig. \ref{EX4}.
 \begin{figure}
  \includegraphics[width=\linewidth]{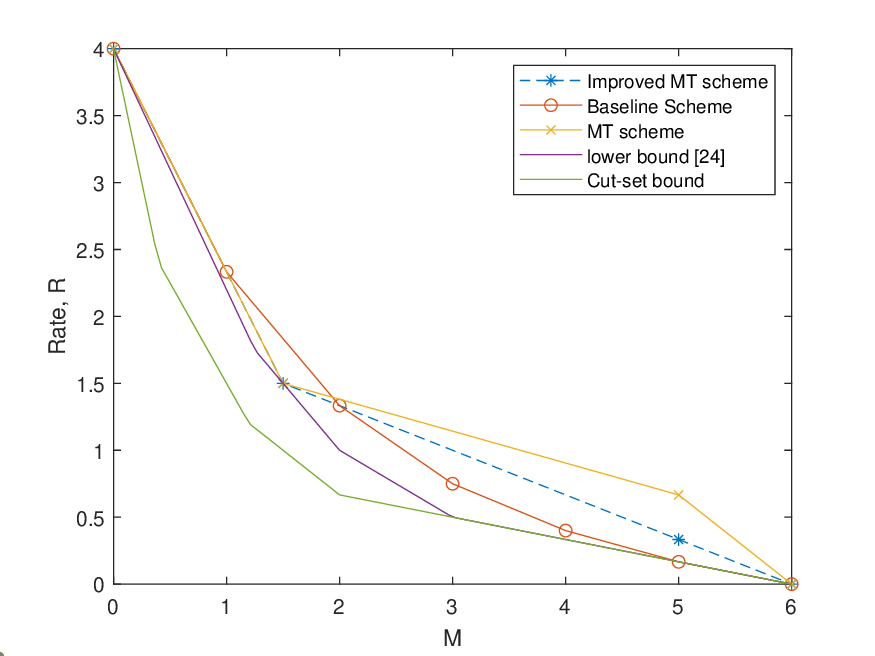}
  \caption{Example 4: (6, 4, 6) hotplug coded caching system.}
  \label{EX4}
\end{figure}

\item In Example \ref{ex3}, the set $T=\{1,3,4,5,7,8\}$ satisfies the condition given in (\ref{T}), and therefore we have
$$\overline{B}=\begin{bmatrix}
* &  &  & 6 & * \\
* & * & 2 &  & \\
 & * & * &  & \\
& 2 & * & * & 9\\
6 &  &  & * & *
\end{bmatrix},  $$
Now by using Algorithm \ref{algo1} for HpPDA $(\overline{B}, P)$, we get rate $R=\frac{3}{5}$.
That means in the delivery phase, the server only transmits $X_2, X_6$ and $X_9$. The user $4$ gets coded subfile $C_{1,4}$ using $X_2$. Then the user $4$ gets its desired file $W_1$ using $C_{1,5}, C_{1,9}, C_{1,10}, C_{1,12}$ from its cache and $C_{1,4}$. 
\end{itemize}

\end{remark}

For some other examples, the comparison of the improved MT scheme with existing schemes is given in Fig. \ref{OEX1} and Fig. \ref{OEX2}.

\begin{figure}
  \includegraphics[width=\linewidth]{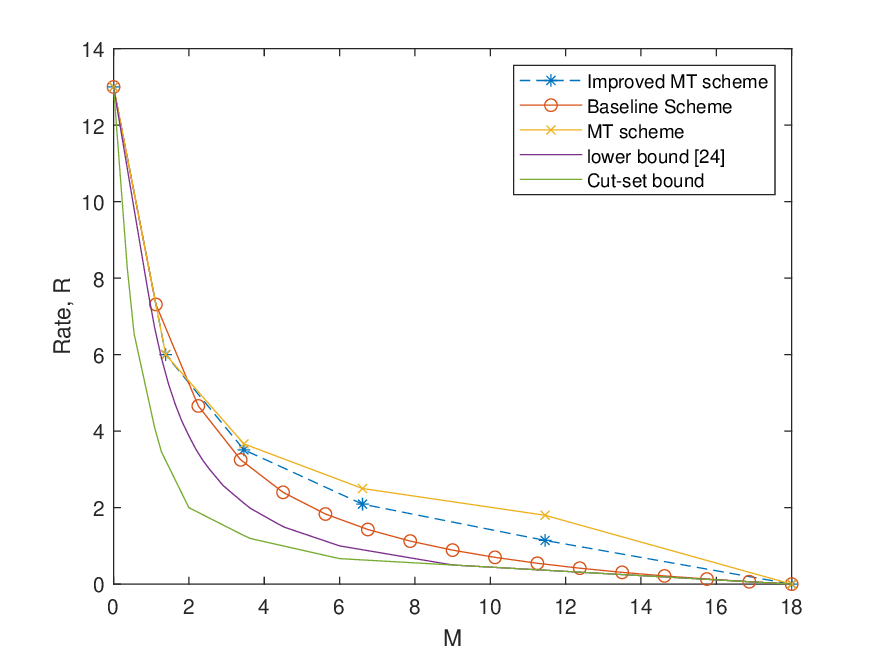}
  \caption{ (16, 13, 18) hotplug coded caching system.}
  \label{OEX1}
\end{figure}
\begin{figure}
  \includegraphics[width=\linewidth]{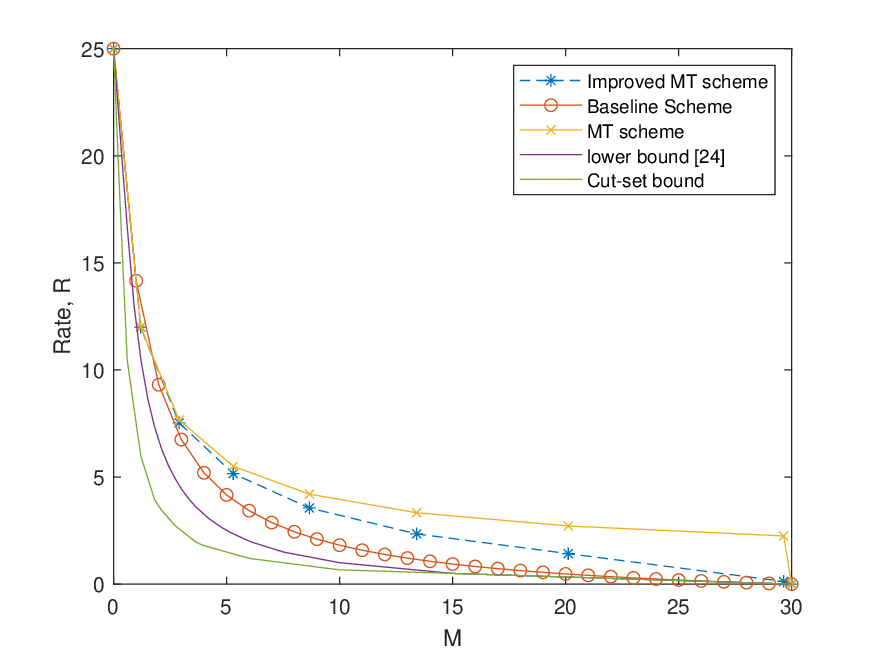}
  \caption{(30, 25, 30) hotplug coded caching system.}
  \label{OEX2}
\end{figure}

\section{Construction of HpPDA from $t$-designs} \label{tdesignPDA}


Let $(X,\mathcal{A})$ be a $t$-$(v,k,\lambda)$ design with non-repeated blocks, where $X=\{1,2,\ldots, v\}$, $\mathcal{A}=\{A_1, A_2, \ldots, A_b\}$ and $|A_i|=k$ for all $i \in \{1,2, \ldots, b\}$. Consider an array $P$ whose rows are indexed by the blocks in $\mathcal{A}$ and columns are indexed by the points in $X$. The array $P=(P(A, i))_{A \in \mathcal{A}, i \in X}$ is a $b \times v$ array containing ``$*$" and null, and defined as
\begin{equation} \label{arrayP}
P(A,i)=\begin{cases}
* & \text{if} \ \ i \in A, \\
null & \text{if} \ \ i \notin A
\end{cases}.
\end{equation}
In other words, $P$ can be obtained by the transpose of the incidence matrix of design $(X,\mathcal{A})$ after replacing $1$ by $*$, and $0$ by $null$.

For $1 \leq s \leq t$, consider the parameters $\lambda_s$ and $\lambda_s^{t}$ 
given in \eqref{lambda-s} and \eqref{lambda-ts}, respectively. Now for $1 \leq s \leq t-1$, let $0 \leq a_s \leq \lambda_s^{t}$, and consider a set 
$$\mathcal{\mathcal{R}}=\bigcup_{s=1}^{t-1} \left\{ (Y, i) \mid Y \in {[t] \choose s}, i \in [a_s] \right\}.$$

Clearly, $|\mathcal{R}|=\sum_{s=1}^{t-1} a_s {t \choose s}$. The integers $a_s\in \{0,1, \ldots, \lambda_s^t\}, s \in [t-1]$ are choosen in such a way that $|\mathcal{R}| > \lambda_1$. 
Consider an array $B$ whose rows are indexed by the elements in $\mathcal{R}$ and columns are indexed by the points in $[t]$. The array $B=(B((Y, i), j))$ is a $|\mathcal{R}| \times t$ array, defined as
\begin{equation} \label{arrayB}
B((Y, i), j)=\begin{cases}
* & \text{if} \ \ j \in Y, \\
\left(Y\cup\{j\},i\right) & \text{if} \ \ j \notin Y
\end{cases}.
\end{equation}

\noindent \textbf{Note:} Further, in this section, non star entries in an array (or PDA) are refered as integers. For example, in array $B$ defined in \eqref{arrayB}, the non star entries are of the form $\left(Y', i\right)$, where $Y'$ is a subset of $[t]$ of size $s+1$ and $i \in [t]$, which will be refered as the integer entries.

\begin{example}\label{ex6}
Consider the $3$-$(8,4,1)$ design $(X, \mathcal{A})$, where \\
$
X=\{1,2,3,4,5,6,7,8\},$ \\
$\mathcal{A}=\{ 1256, 3478, 2468, 1357, 1458,2367,1234, 5678, 1278,$ $3456, 1368, 2457, 1467, 2358\}.$\\
The $14 \times 8$ array $P$ can be obtained as follows.

\[
P=\begin{blockarray}{ccccccccc}
& 1 & 2 & 3 & 4 & 5 & 6 & 7 & 8 \\
\begin{block}{c[cccccccc ]}
1234 & *& * & * & * &  &  &  & \\
1256 &* & * &  &  & * & * & & \\
1278 & *& * &  &  &  &  & * & *\\
1357 & * &  & * &  & * &  & * & \\
1368 & *&  &*  &  &  &*  &  & *\\
1458 & *&  &  & * & * &  &  & *\\
1467 & *&  &  & * &  & * & * & \\
3478 & &  & * & * &  &  & * & *\\
2468 & & * &  & * &  & * &  & *\\
2358 & & * & * &  & * &  &  & *\\
2367 & & * & * &  &  &*  & * & \\
2457 & & * &  &*  & * &  & * & \\
3456 & &  & * & * & * &*  &  & \\
5678 & &  &  &  & * & * & * & *\\
\end{block}
\end{blockarray}.
 \]

In this example, we have $\lambda_2^3=\lambda_1^3=2$. By choosing $a_2=2$ and $a_1=1$, we have $|\mathcal{R}|=9$ and
\begin{align*}
\mathcal{R} & =  \left\{ (Y, i) \mid Y \in {[3] \choose 2}, i \in [2] \right\} \\
&\cup  \left\{ (Y, i) \mid Y \in {[3] \choose 1}, i \in [1] \right\} \\
& =\left\{ (\{1,2\}, 1), (\{1,3\}, 1), (\{2,3\}, 1), (\{1,2\}, 2), \right.\\
&  \left. (\{1,3\}, 2),  (\{2,3\}, 2), (\{1\}, 1), (\{2\}, 1), (\{3\}, 1) \right\}.
\end{align*}
Therefore, the $9 \times 3$ array $B$ can be obtained as follows.
\[
B=\begin{blockarray}{cccc}
& 1 & 2 & 3  \\
\begin{block}{c[ccc ]}
(12, 1) &* & * &  (123, 1) \\
 (13, 1) & * & (123, 1) & *\\
 (23, 1) &(123, 1) & * & *\\
(12, 2) &* & * & (123, 2)  \\
 (13, 2) & * & (123, 2) & *\\
 (23, 2) & (123, 2) & * & *\\
(1, 1) &* & (12, 1) &  (13, 1) \\
 (2, 1) & (12, 1) & * & (23, 1) \\
 (3, 1) &(13, 1) & (23, 1) & *\\
\end{block}
\end{blockarray}
 \]
or equivalently,
\[
B=
\begin{blockarray}{ccc}
 & &  \\
\begin{block}{[ccc]}
* & * &  1 \\
  * & 1 & *\\
 1 & * & *\\
* & * & 2  \\
 * &2 & *\\
  2 & * & *\\
* & 3 &  4 \\
 3 & * & 5 \\
4 & 5 & *\\
\end{block}
\end{blockarray}.
\]

\end{example}

To show that the arrays $P$ and $B$ constructed above using a $t$-design form a HpPDA, first we prove that $B$ is a PDA in the following lemma.
\begin{lemma}\label{lemma1}
For $s \in [t-1]$, and $a_s\in \{0,1, \ldots, \lambda_s^t\}, $, the array $B$ defined in \eqref{arrayB} is a $[K', F', Z', S]$ PDA, where 
\begin{align}\label{parameters1}
 \nonumber K'&=t, F'=|\mathcal{R}|=\sum_{s=1}^{t-1} a_s {t \choose s}, 
Z'=\sum_{s=1}^{t-1} a_s {t-1 \choose s-1},\\ 
 &\text{and} \ S=\sum_{s=1}^{t-1} a_s {t \choose s+1}.
\end{align}
\end{lemma}

\begin{proof}
The array $B$ is a $|\mathcal{R}| \times t$ array whose rows are indexed by the elements in $\mathcal{R}$ and columns are indexed by the points in $[t]$. We can rewite $\mathcal{R}$ as 
$$\mathcal{R}=\bigcup_{s=1}^{t-1} \bigcup_{i \in [a_s]}\left\{ (Y, i) \mid Y \in {[t] \choose s} \right\}=\bigcup_{s=1}^{t-1} \bigcup_{i \in [a_s]}\mathcal{R}_{(s,i)},$$
where $\mathcal{R}_{(s,i)} =\left\{ (Y, i) \mid Y \in {[t] \choose s} \right\}$ for all $s\in [t-1]$ and $i \in [a_s]$. Clearly, $|\mathcal{R}_{(s,i)}|={t \choose s}$. Now break the array $B$ into $\sum_{s=1}^{t-1}a_s$ subarrays denoted by $B_{(s,i)}$, where $s\in [t-1], i \in [a_s]$. The subarray $B_{(s,i)}$ is a ${t \choose s} \times t$ array whose rows are indexed by the elements in $\mathcal{R}_{(s,i)}$ and columns are indexed by the elements in $[t]$.

To prove that $B$ is a PDA, we will prove that $B_{(s,i)}$ is a PDA, for all $s\in [t-1], i \in [a_s]$, and there is no integer common in two different subarrays $B_{(s,i)}$ and $B_{(s',i')}$, where either $s \neq s'$ or $i \neq i'$. 
For all $s\in [t-1], i \in [a_s]$, we have 
\begin{equation*} \label{arrayBsi}
B_{(s,i)}((Y, i), j)=\begin{cases}
* & \text{if} \ \ j \in Y, \\
\left(Y\cup\{j\},i\right) & \text{if} \ \ j \notin Y
\end{cases}.
\end{equation*}

\begin{enumerate}[{C}1.]
\item The number of ``$*$"'s appear in $j$-th column of $B_{(s,i)}$ is 
$$\left| \left\{Y \in {[t] \choose s} \mid j \in Y \right\} \right|={t-1 \choose s-1},$$ for all $j \in [t].$
\item The set of integers appearing in array $B_{(s,i)}$ is
$$S=\left\{ (Y \cup \{j\}, i) \mid Y \in {[t] \choose s}, j \in [t] \right\}$$
 (here, integers are denoted by a pair $(U,i)$, where $U$ is set of size $s+1$ and $i \in [t]$). Clearly, $|S|={t \choose s+1}$, and each integer appears exactly $s+1$ times.
\item Consider two distinct integer entries $B_{(s,i)}((Y_1,i), j_1)$ and $B_{(s,i)}((Y_2,i), j_2)$ such that $B_{(s,i)}((Y_1,i), j_1)=B_{(s,i)}((Y_2,i), j_2) = Q$, where $Q \in S$. We have
$Q=(Y_1\cup\{j_1\},i)=(Y_2\cup\{j_2\},i)$ which implies that 
\begin{equation}\label{leq1}
Y_1\cup\{j_1\}=Y_2\cup\{j_2\}.
\end{equation} 
Since both entries are different, we have $(Y_1, i) \neq (Y_2, i)$ and $j_1 \neq j_2$. Using \eqref{leq1}, and the fact that $j_1 \neq j_2$, we have $j_1 \in Y_2$ and $j_2 \in Y_1$, which implies that $B_{(s,i)}((Y_1,i), j_2)=B_{(s,i)}((Y_2,i), j_1) = *$.

\end{enumerate}
Therefore, $B_{(s,i)}$ is a $(s+1)$-$[t, {t \choose s}, {t-1 \choose s-1}, {t \choose s+1}]$ regular PDA for all $s\in [t-1], i \in [a_s]$.
Since every integer entry $(U,i)$ in PDA $B_{(s,i)}$ depends on $i$, there is no integer common in PDA $B_{(s,i)}$ and $B_{(s,i')}$, where $i \neq i'$. Also in an integer entry $(U,i)$ in PDA $B_{(s,i)}$, the cardinality of $U$ is $s+1$, hence no integer entry will be common in PDA $B_{(s,i)}$ and $B_{(s',i)}$, where $s \neq s'$.

The array $B$ is made of $\sum_{s=1}^{t-1}a_s$ disjoint subarrays $B_{(s,i)}$, $s\in [t-1], i \in [a_s]$, and each subarrays $B_{(s,i)}$ is a $[t, {t \choose s}, {t-1 \choose s-1}, {t \choose s+1}]$ PDA. Therefore, $B$ is a $\left[ K'=t, F'=\sum_{s=1}^{t-1} a_s {t \choose s}, Z'=\sum_{s=1}^{t-1} a_s {t-1 \choose s-1}, \right. $ $\left. S=\sum_{s=1}^{t-1} a_s {t \choose s+1} \right] \text{PDA}.$

\end{proof}

\begin{remark}
For some $s \in [t-1]$ and $i \in [a_s]$, the PDA $B_{(s, i)}$ is a $(s+1)$-$[t, {t \choose s}, {t-1 \choose s-1}, {t \choose s+1}]$ MAN PDA.
\end{remark}

We identify all PDAs $B_{(s, i)}, s \in [t-1], i \in [a_s]$ in Example \ref{ex6} as follows.
\begin{example*}[Example \ref{ex6} continued]

For $s \in [2], i \in [a_s]$ we have the following PDAs.
\[
B_{(1, 1)}=\begin{blockarray}{cccc}
& 1 & 2 & 3  \\
\begin{block}{c[ccc ]}
(1, 1) &* & (12, 1) &  (13, 1) \\
 (2, 1) & (12, 1) & * & (23, 1) \\
 (3, 1) &(13, 1) & (23, 1) & *\\
\end{block}
\end{blockarray},
 \]
\[
B_{(2, 1)}=\begin{blockarray}{cccc}
& 1 & 2 & 3  \\
\begin{block}{c[ccc ]}
(12, 1) &* & * &  (123, 1) \\
 (13, 1) & * & (123, 1) & *\\
 (23, 1) &(123, 1) & * & *\\
\end{block}
\end{blockarray},
 \]
\[
B_{(2,2)}=\begin{blockarray}{cccc}
& 1 & 2 & 3  \\
\begin{block}{c[ccc ]}
(12, 2) &* & * & (123, 2)  \\
 (13, 2) & * & (123, 2) & *\\
 (23, 2) & (123, 2) & * & *\\
\end{block}
\end{blockarray}.
 \]

\end{example*}

\begin{theorem}\label{thm3}
For $a_s\in \{0,1, \ldots, \lambda_s^t\}, s \in [t-1]$, such that $|\mathcal{R}|> \lambda_1$, the pair  $(P, B)$ forms a $(K,K',F, F',Z,Z',S)$-HpPDA, where $K', F', Z', S$ are defined in \eqref{parameters1} and $K=v, F=b, Z=\lambda_1$.
\end{theorem}

\begin{proof}
From Lemma \ref{lemma1}, we know that $B$ is a $\left[ K'=t, F'=\sum_{s=1}^{t-1} a_s {t \choose s}, Z'=\sum_{s=1}^{t-1} a_s {t-1 \choose s-1}, S=\sum_{s=1}^{t-1} \right.$ $\left. a_s {t \choose s+1} \right]$ PDA. The array $P$ defined in \eqref{arrayP} is a $b \times v$ array containing ``$*$" and null, whose rows are indexed by the blocks in $\mathcal{A}$ and columns are indexed by the points in $X$, i.e., $F=b, K=v$. From Theorem \ref{thm1.1}, we know that in a $t$-$(v,k,\lambda)$ design each point appears in exactly $\lambda_1$ blocks. Hence there are exactly $\lambda_1$ number of ``$*$"\textquotesingle s in each column of $P$, i.e., $Z=\lambda_1$.
 
Now we need to prove that the pair $(P, B)$ satisfies the property \eqref{c}.

For a given subset $\tau \subseteq X, |\tau|=t$ and $U \subseteq \tau, |U|=s<t$, there are exactly $\lambda_s^t$ blocks which contains $U$ and does not contain any element from $\tau \backslash U$ (from Theorem \ref{thm1.2}). Let the set of those blocks be denoted by $\mathcal{A}_{U}^{\tau}$, i.e., $\mathcal{A}_{U}^{\tau} = \{A\in \mathcal{A} \mid A \cap \tau =U\}$. Clearly, $|\mathcal{A}_{U}^{\tau}|=\lambda_s^t$. Let the elements in $\mathcal{A}_{U}^{\tau}$ are denoted as $\{\mathcal{A}_{U}^{\tau} (1), \mathcal{A}_{U}^{\tau} (2), \ldots, \mathcal{A}_{U}^{\tau}(\lambda_s^t)\}$. In other words, we can say that $\mathcal{A}_{U}^{\tau} (i)$ is a block which contains $U$ and does not contain any element from $\tau \backslash U$.

Let $\tau=\{i_1,i_2, \ldots, i_{t}\} \subseteq X$. Consider a set $\zeta = \bigcup_{s=1}^{t-1} \bigcup_{i \in [a_s]}\left\{ \mathcal{A}_{U}^{\tau} (i) \mid U \in {\tau \choose s} \right\}$ $=\bigcup_{s=1}^{t-1}$ $\bigcup_{i \in [a_s]} \zeta_{(s,i)}$, where $\zeta_{(s,i)}=\left\{ \mathcal{A}_{U}^{\tau} (i) \mid U \in {\tau \choose s} \right\}$. Clearly, $|\zeta|=\sum_{s=1}^{t-1} a_s {t \choose s}=F'$. 
Corresponding to the sets $\tau$ and $\zeta$, we have a subarray $\overline{P}=[P]_{ \zeta\times \tau}$ of $P$ such that $\overline{P}(A, i_j)=*$ if and only if $i_j \in A, A\in \zeta$. Now divide array $\overline{P}$ into $\sum_{s=1}^{t-1}a_s$ subarrays denoted by $\overline{P}_{(s,i)}$, for $s\in [t-1], i \in [a_s]$. The subarray $\overline{P}_{(s,i)}$ is a ${t \choose s} \times t$ array whose rows are indexed by the elements in $\zeta_{(s,i)}$ and columns are indexed by the elements in $\tau$ such that $\overline{P}_{(s,i)}(A, i_j)=*$ if and only if $i_j \in A, A\in \zeta_{(s,i)}$, or we can say $\overline{P}_{(s,i)}(\mathcal{A}_{U}^{\tau} (i), i_j)=*$ if and only if $i_j \in \mathcal{A}_{U}^{\tau} (i), U \in {\tau \choose s}$. Since $\mathcal{A}_{U}^{\tau} (i)$ denotes a block which contains $U$ and does not contain any element from $\tau \backslash U$, we can say that $\overline{P}_{(s,i)}(\mathcal{A}_{U}^{\tau} (i), i_j)=*$ if and only if $i_j \in U, U \in {\tau \choose s}$.

To prove that $\overline{P} \myeq B$, we will prove that $\overline{P}_{(s,i)} \myeq B_{(s,i)}$ for all $s\in [t-1]$ and $i \in[a_s]$. 

The array $B_{(s,i)}$ is a ${t \choose s} \times t$ array whose rows are indexed by the elements in $\mathcal{R}_{(s,i)}$ and columns are indexed by the elements in $[t]$ such that $B_{(s,i)}((Y, i), j)=*$ if and only if $j \in Y$, where $\mathcal{R}_{(s,i)} =\left\{ (Y, i) \mid Y \in {[t] \choose s} \right\}$ for all $s\in [t-1]$ and $i \in [a_s]$. We prove $\overline{P}_{(s,i)} \myeq B_{(s,i)}$ in the following three steps. 
\begin{enumerate}
\item First, we show that there is an one to one correspondence between the sets which are used to index the columns of arrays $\overline{P}_{(s,i)}$ and $B_{(s,i)}$. Clearly, there is an one to one correspondence $\phi : \tau \to [t]$ such that $\phi(i_j)=j, \forall i_j \in \tau$.
\item In this step, we show that there is an one to one correspondence between the sets which are used to index the rows of arrays $\overline{P}_{(s,i)}$ and $B_{(s,i)}$. For given $s\in [t-1]$ and $i \in[a_s]$, there is an one to one correspondence $\psi: \zeta_{(s,i)} \to \mathcal{R}_{(s,i)} $ such that $\psi(\mathcal{A}_{U}^{\tau} (i))=(\{j_1, j_2, \ldots, j_s\},i)$, for all $U=\{i_{j_1}, i_{j_2}, \ldots, i_{j_s}\}\subseteq \tau, |U|=s$.
\item In the last step, we show that the conditions to have a "*" entry in $\overline{P}_{(s,i)}$ and $B_{(s,i)}$ corresponds to each other under the defined maps $\phi$ and $\psi$. Clearly, for $1 \leq j \leq t$, we have $i_j \in U$, $U \in {\tau \choose s}$ if and only if $\phi(i_j) \in Y$, where $\psi(\mathcal{A}_{U}^{\tau} (i))=(Y,i)$. Therefore, we have 
$\overline{P}_{(s,i)}(\mathcal{A}_{U}^{\tau} (i), i_j)=B_{(s,i)}(\psi(\mathcal{A}_{U}^{\tau} (i)), \phi(i_j))=*$ if and only if $i_j \in U, U \in {\tau \choose s}$. 
\end{enumerate}
Hence $\overline{P}_{(s,i)} \myeq B_{(s,i)}$.
\end{proof}

\begin{remark}
The users are represented by the points in $X$ and subpackitization level is $F'$. All $F'$ subfiles of each file are coded with $[F, F']$ MDS code to get $F$ coded subfiles. 
\end{remark}

The following result can be obtained directly from Theorem \ref{thm3} and Theorem \ref{thm2.1}.

\begin{corollary}\label{thm4}[$t$-Scheme]
For $a_s\in \{0,1, \ldots, \lambda_s^t\}, s \in [t-1]$, such that $|\mathcal{R}| > \lambda_1$, HpPDA $(P, B)$ gives a $(K, K', N)$ hotplug coded caching scheme with cache memory $M$ and rate ${R}$ as follows
$$\frac{M}{N}=\frac{Z}{F'}=\frac{\lambda_1}{\sum_{s=1}^{t-1} a_s {t \choose s}} \ \text{and} \ R=\frac{S}{F'}=\frac{\sum_{s=1}^{t-1} a_s {t \choose s+1}}{\sum_{s=1}^{t-1} a_s {t \choose s}}.$$
\end{corollary}

\begin{example*}[Example \ref{ex6} continued]

Consider a subset of $[8]$ of size $3$, say $\tau=\{2,6,8\}$. For all $U \subseteq \tau, |U|=2$, we have $\mathcal{A}_{\{2,6\}}^\tau=\{1256, 2367\}$, $\mathcal{A}_{\{2,8\}}^\tau=\{1278, 2358\}$ and $\mathcal{A}_{\{6,8\}}^\tau=\{5678, 1368\}$. For all $U \subseteq \tau, |U|=1$, we have $\mathcal{A}_{\{2\}}^\tau=\{1234, 2457\}$, $\mathcal{A}_{\{6\}}^\tau=\{3456, 1467\}$ and $\mathcal{A}_{\{8\}}^\tau=\{3478, 1458\}$.

Therefore, there exists a subset of $\mathcal{A} $, 
\begin{align*}
\zeta &= \bigcup_{s=1}^{t-1} \bigcup_{i \in [a_s]}\left\{ \mathcal{A}_{U}^{\tau} (i) \mid U \in {\tau \choose s} \right\} \\
&=\{1256, 1278, 5678, 2367,  2358, 1368, 1234, 3456, 3478\} 
\end{align*}
 such that we have a subarray $[P]_{\zeta \times \tau}$ of $P$ as follows.

\[
[P]_{\zeta \times \tau}=\begin{blockarray}{ccccc}
& 2 &   6  & 8 & \\
\begin{block}{c[ccc ]c}
 1256 &* & * & & (26,1)\\
 1278 & *&  &  * & (28,1) \\
  5678 & & * & * & (68,1)\\
  2367 &* & * & & (26,2)\\
 2358 &* & &*  & (28,2) \\
  1368 & & * & * & (68,2)\\
 1234 & *  &  & & (2,1)  \\
3456 & & * & & (6,1)    \\
 3478 & &  & * & (8,1) \\
\end{block}
\end{blockarray} \myeq B.
 \]

Here, $(P, B)$ is an $(8,3, 14, 9, 7, 5, 5)$ HpPDA which corresponds to $(8,3,N)$ hotplug coded caching system with 
$\frac{M}{N}=\frac{7}{9} \ \text{and} \ R=\frac{5}{9}.$
Since $Z-Z'=2$, the rate can be reduced further to $\frac{2}{9}$ by using the method given in Section \ref{improved} (explained in deatil in the next section).

For conditions, $0 \leq a_2 \leq 2, 0 \leq a_1 \leq 2$, and $|\mathcal{R}|>\lambda_1=7$, we have following three choices for $a_1$ and $a_2$,
\begin{enumerate}
\item $a_2=2, a_1=1$ with $|\mathcal{R}|=9$,
\item $a_2=1, a_1=2$ with $|\mathcal{R}|=9$,
\item $a_2=2, a_1=2$ with $|\mathcal{R}|=12$.
\end{enumerate}

First case (Case 1)) has already been discussed. Now we explain other two cases.

\noindent Case 2) For $a_2=1$ and $a_1=2$, we have $|\mathcal{R}|=9$ and
\begin{align*}
\mathcal{R} & =  \left\{ (\{1,2\}, 1), (\{1,3\}, 1), (\{2,3\}, 1),  (\{1\}, 1), (\{2\}, 1),\right.\\
&  \left.  (\{3\}, 1),  (\{1\}, 2), (\{2\}, 2), (\{3\}, 2) \right\}.
\end{align*}
Therefore, the $9 \times 3$ array $B'$ can be obtained as follows.
\[
B'=\begin{blockarray}{cccc}
& 1 & 2 & 3  \\
\begin{block}{c[ccc ]}
(12, 1) &* & * &  (123, 1) \\
 (13, 1) & * & (123, 1) & *\\
 (23, 1) &(123, 1) & * & *\\
(1, 1) &* & (12, 1) &  (13, 1) \\
 (2, 1) & (12, 1) & * & (23, 1) \\
 (3, 1) &(13, 1) & (23, 1) & *\\
(1, 2) &* & (12, 2) &  (13, 2) \\
 (2, 2) & (12, 2) & * & (23, 2) \\
 (3, 2) &(13, 2) & (23, 2) & *\\
\end{block}
\end{blockarray}.
 \]
Then we get $(8,3, 14, 9, 7, 4, 7)$ HpPDA $(P, B')$ which corresponds to $(8,3,N)$ hotplug coded caching system with 
$\frac{M}{N}=\frac{7}{9} \ \text{and} \ R=\frac{7}{9}.$
Since $Z-Z'=3$, the rate can be reduced further to $\frac{3}{9}$ by using the method given in Section \ref{improved}.

\noindent Case 3) For $a_2=2$ and $a_1=2$, we have $|\mathcal{R}|=12$ and
\begin{align*}
\mathcal{R} & =  \left\{ (\{1,2\}, 1), (\{1,3\}, 1), (\{2,3\}, 1), (\{1,2\}, 2), (\{1,3\}, 2), \right.\\
&  \left.  \qquad (\{2,3\}, 2),  (\{1\}, 1), (\{2\}, 1), (\{3\}, 1),  (\{1\}, 2), (\{2\}, 2), \right.\\
&  \left. \qquad \quad  (\{3\}, 2) \right\}.
\end{align*}
Therefore, the $12 \times 3$ array $B''$ can be obtained as follows.
\[
B''=\begin{blockarray}{cccc}
& 1 & 2 & 3  \\
\begin{block}{c[ccc ]}
(12, 1) &* & * &  (123, 1) \\
 (13, 1) & * & (123, 1) & *\\
 (23, 1) &(123, 1) & * & *\\
(12, 2) &* & * &  (123, 2) \\
 (13, 2) & * & (123, 2) & *\\
 (23, 2) &(123, 2) & * & *\\
(1, 1) &* & (12, 1) &  (13, 1) \\
 (2, 1) & (12, 1) & * & (23, 1) \\
 (3, 1) &(13, 1) & (23, 1) & *\\
(1, 2) &* & (12, 2) &  (13, 2) \\
 (2, 2) & (12, 2) & * & (23, 2) \\
 (3, 2) &(13, 2) & (23, 2) & *\\
\end{block}
\end{blockarray}.
 \]
Here we will get $(8,3, 14, 12, 7, 6, 8)$ HpPDA $(P, B'')$ which corresponds to $(8,3,N)$ hotplug coded caching system with 
$\frac{M}{N}=\frac{7}{12} \ \text{and} \ R=\frac{8}{12}=\frac{2}{3}.$ Again, since $Z-Z'=1$, the rate can be reduced further to $\frac{7}{12}$ by using the method given in Section \ref{improved}.

The parameters of hotplug coded caching scheme for all three choices of $a_i$'s using $3$-$(8,4,1)$ design are given in Table \ref{tab2}.

\begin{table}[!t]
\caption{Example \ref{ex6}:  The parameters of $(8,3,N)$ hotplug coded caching scheme for all three choices of $a_i$'s. \label{tab2}}
\centering
\begin{tabular}{|c|c|c|c|c|} 
 \hline
 $(a_2, a_1)$ & $|\mathcal{R}|$ & $M/N$ & Rate & Improved rate \\ 
 \hline
$(2, 1)$ & $9$ & $7/9$ & $5/9$ & $2/9$ \\

$(1, 2)$ & $9$ & $7/9$ & $7/9$ & $3/9$ \\

$(2, 2)$ & $12$ & $7/12$ & $8/12$ & $7/12$ \\
\hline
\end{tabular}
\end{table}

\end{example*}

For a given $t$-$(v,k,\lambda)$ design, we get a $(v,t,N)$ hotplug coded caching scheme. We get different cache memory points for the different choice of $a_s, s \in [t-1]$ such that $1 \leq a_s \leq \lambda_s^t$ and $|\mathcal{R}| > \lambda_1$. In Example \ref{ex6}, we get $(8,3,N)$ hotplug coded caching scheme for $\frac{M}{N}=\frac{7}{9}$ (from Case 1) and Case 2)) and $\frac{M}{N}=\frac{7}{12}$ (from Case 3)). Suppose we get the same value of $|\mathcal{R}|$ for different choices of $a_s, s \in [t-1]$, i.e., $\frac{M}{N} =\frac{Z}{F'}$ is same as in Case 1) and Case 2) of Example \ref{ex6}. Then we simply consider the choice of $a_s, s \in [t-1]$ with minimum value of $\sum_{s=1}^{t-1} a_s {t \choose s+1}$ which will correspond to the minimum rate.

\begin{remark}\label{rem6}
Since a $t$-$(v, k, \lambda)$ design $(X, \mathcal{A})$ is also an $s$-$(v, k, \lambda_s)$ design, for all $1 \leq s \leq t$, therefore, a $(v, s, N)$ hotplug coded caching scheme can be constructed with $s$ number of active users using the proposed scheme for $s$-$(v, k, \lambda_s)$ design.
\end{remark}

\begin{remark}
Since there exists a $(t-i)$-$(v-i,k-i,\lambda)$ design if there is a $t$-$(v,k,\lambda)$ design $(X, \mathcal{A})$, for $i<t$, therefore, a $(v-i, t-i, N)$ hotplug coded caching scheme can be constructed with $v-i$ users and $t-i$ active users using the proposed $t$-Scheme for $(t-i)$-$(v-i,k-i,\lambda)$ design.
\end{remark}

\subsection{Improved version of $t$-Scheme}\label{improvescheme2}

Consider the HpPDA $(P, B)$ for $t$-Scheme constructed in Section \ref{tdesignPDA}, where $B$ and $P$ are defined in (\ref{arrayB}) and (\ref{arrayP}). Now, we will reduce the rate of $t$-Scheme using the method given in Section \ref{improved} in which we choose a subset $T \subseteq S$ such that 
\begin{equation}\label{T'}
|T \cap C_i| \leq Z-Z', \ \forall i \in [K'],
\end{equation} 
where $C_i$ denotes the set of integers appearing in the $i$-th column of array $B$, $i \in [K']$. 

There are total $Z'=\sum_{s=1}^{t-1} a_s {t-1 \choose s-1}$ number of $*$'s and $F'-Z'=\sum_{s=1}^{t-1} a_s {t \choose s} - \sum_{s=1}^{t-1} a_s {t-1 \choose s-1}$ number of integers in each column of array $B$.
Since for all $s\in [t-1]$ and $i \in [a_s]$, $B_{(s,i)}$ is a $[t, {t \choose s}, {t-1 \choose s-1}, {t \choose s+1}]$ regular PDA in which each integer occurs exactly $s+1$ times. So to find a set $T\subset[S]$, which satisfies the above condition, we start choosing integers from subarray $B_{(s,i)}$ for the minimum value of $s$ for which $a_s\neq 0$. Since the occurrence of an integer in $B_{(s, i)}$ will be less than the occurrence of an integer in $B_{(s', i)}$ if $s<s'$, we can choose a set $T$ with larger cardinality.

The following lemma will help to find a subset of the set of integers in $B_{(s,i)}$ which satisfies condition \eqref{T'}.

\begin{lemma} \label{lemma2}
For a $(s+1)-\left[ t, {t \choose s}, {t-1 \choose s-1}, {t \choose s+1} \right]$ regular PDA $B_{(s,i)}$, for some $s\in [t-1]$ and $i \in [a_s]$, defined in the proof of Lemma \ref{lemma1}, and a positive integer $z$, there exist a set $T \subseteq S$ such that
$$|T \cap C_{B_{(s,i)}}(j)| \leq z, \ \forall j \in [t],$$
where $S$ is the set of all integers in $B_{(s, i)}$ and $C_{B_{(s,i)}}(j)$ denotes the set of integers appearing in $j$-th column of $B_{(s,i)}$. The cardinality of the set $T$ is
$$|T|=\begin{cases}
\left \lfloor \frac{t}{s+1} \right \rfloor z, & \text{if} \ z < {t-1 \choose s}, \\
{t \choose s+1}, & \text{if} \ z \geq {t-1 \choose s}.
\end{cases}$$
(A way to find a set $T$ with given cardinality, when $z < {t-1 \choose s}$, is presented in Function 1.)
\end{lemma}

\begin{proof}
If $z \geq {t-1 \choose s}$ then take $T=S$. Since the number of integers in a column of $B_{(s,i)}$ is ${t \choose s} - {t-1 \choose s-1}={t-1 \choose s}$, we have
\begin{equation}\label{T''}
|T \cap C_{B_{(s,i)}}(j)| = {t-1 \choose s} \leq z, \ \forall j \in [t],
\end{equation}
and clearly, $|T|={t \choose s+1}$. Now consider the case when $z < {t-1 \choose s}$.
Each element $(Y \cup \{j\},i)$ in $S$ corresponds to a subset $Y'=Y \cup \{j\}$ of $[t]$ of size $s+1$, and corresponding to each subset $Y' \subseteq [t]$ of size $s+1$, there is an element $(Y', i)$ in the set $S$, and $(Y', i)$ will appear in $B_{(s,i)}$ at the positions
$B_{(s,i)} ((Y'\backslash \{y\}, i), y) \ \text{for all} \ y \in Y',$
i.e., integers corresponding to a set $Y'$ of size $s+1$ will appear in the column $y$, for all $y \in Y'$. Consider a set $T'= \{(Y', i) \mid Y' \subseteq [t], |Y'| =s+1\}$ such that $Y'_1 \cap Y'_2 =\emptyset$ for $Y'_1, Y'_2 \in T'$ and $Y'_1 \neq Y'_2$. The maximum possible cardinality of such a set is $\left \lfloor \frac{t}{s+1} \right \rfloor$; let $|T'|=\left \lfloor \frac{t}{s+1} \right \rfloor$. Clearly, $|T' \cap C_{B_{(s,i)}}(j)| \leq 1, \ \forall j \in [t].$
Now consider set $T$ as the union of different $z$ number of sets like $T'$. Then
$$|T \cap C_{B_{(s,i)}}(j)| \leq z, \ \forall j \in [t],$$
and $|T|=z\left \lfloor \frac{t}{s+1} \right \rfloor.$

\end{proof}

\begin{algorithm}
\begin{algorithmic}[1]
\Procedure{Function 1}{$s,z$}
 \For{$j \in [1: z ]$}
\State $E \gets [t], \ T'(j) \gets \emptyset$.
\While{$|E| \geq s+1$}
\State Choose a set $U \subseteq E$ such that $|U|=s+1$ 
\State $T'(j) \gets T'(j) \cup \{(U, i)\}$.
\State $E \gets E \backslash U$.
\EndWhile
\EndFor 
\State $T=T'(1) \cup T'(2) \cup \ldots \cup T'(z)$.
\EndProcedure
\end{algorithmic}
\end{algorithm}

\begin{note}  In Step 5 of Function 1, a set $U \subseteq E$ such that $|U|=s+1$ will exist, for all $j \in [1:z]$, because $z < {t-1 \choose s}$.
\end{note}

Now, we will construct a set $T$ satisfying condition \eqref{T'} for the array $B$ constructed in \eqref{arrayB}. The following theorem and Algorithm \ref{algo2} give a way to construct such a set $T$. Consider a set $W$ in which all the values of $s \in [t-1]$ for which $a_s \neq 0$ are arranged in increasing order. Let
\begin{equation}\label{W}
W=\{ s \in [t-1] \mid a_s \neq 0 \} = \{ s_1, s_2, \ldots, s_w \},
\end{equation}
such that $s_1 \leq s_2 \leq \cdots \leq s_w \leq t-1$.

\begin{theorem} \label{thm6}
For $[K', F', Z', S]$ PDA $B$ and $F \times K$ array $P$ defined in \eqref{arrayB} and \eqref{arrayP}, respectively, there exists a set $T \subseteq S$ such that
\begin{equation}\label{T2}
|T \cap C_{B}(j)| \leq Z-Z', \ \forall \ j \in [t],
\end{equation}
where $S$ is the set of all integers in $B$ and $C_{B}(j)$ denotes the set of integers appearing in the $j$-th column of $B$, and the cardinality of the set $T$ is defined as follows.

Consider the set $W$ defined in \eqref{W}. If $(a+1) {t-1 \choose s_j} + \sum_{b=1}^{j-1} \alpha_b > Z-Z' \geq a {t-1 \choose s_j} + \sum_{b=1}^{j-1} \alpha_b$, for some $j \in \{1,2, \ldots, w\}$ and $a \in \{0,1, \ldots, a_{s_j}-1 \}$, then
\begin{multline*}
|T|= \sum_{b=1}^{j-1} a_{s_b} {t \choose s_b+1} + a {t \choose s_j+1} + \left \lfloor \frac{t}{s_j+1} \right \rfloor \\ \left( Z-Z' - \sum_{b=1}^{j-1} \alpha_b - a {t-1 \choose s_j} \right),
\end{multline*}
where $\alpha_b = a_{s_b} {t-1 \choose s_b}$.\\ 
(A way to find a set $T$ with given cardinality is presented in Algotithm \ref{algo2}.)
\end{theorem}

\begin{proof}
Here, we will always use notation $T$ for the set satisfying condition \eqref{T2}. Let $S_{(s,i)}$ denotes the set of integers in PDA $B_{(s,i)}$. Clearly, $S=\cup_{s\in [t-1], i \in [a_s]} S_{(s,i)}$. As explained before, we start choosing integers for set $T$ from subarray $B_{(s,i)}$ for the minimum value of $s$ for which $a_s\neq 0$, i.e., first, we will start choosing integers from the PDAs $B_{(s_1, 1)}, B_{(s_1, 2)}, \ldots, B_{(s_1, a_{s_1})}$. Therefore, if $Z-Z' < {t-1 \choose s_1}$ then from Lemma \ref{lemma2}, we have a set $T \subseteq S_{(s,i)}$ with $|T| = \left \lfloor \frac{t}{s_1+1} \right \rfloor (Z-Z')$. Now, if $Z-Z' \geq {t-1 \choose s_1}$ then from Lemma \ref{lemma2}, we can select all ${t \choose s_1+1}$ integers of $B_{(s_1, 1)}$ in set $T$, and than $T$ will contain ${t-1 \choose s_1}$ integers from each column of array $B_{(s,i)}$ (as well as $B$). Then we select remaining from next array $B_{(s_1,2)}$ if $a_{s_1}>1$. That means if $2 {t-1 \choose s_1} > Z-Z' \geq {t-1 \choose s_1}$ then we have
$$|T|=\underbrace{{t \choose s_1+1}}_{\text{from} \ B_{(s_1,1)}} + \underbrace{\left \lfloor \frac{t}{s_1+1} \right \rfloor \left(Z-Z'-{t-1 \choose s_1} \right)}_{\text{from} \ B_{(s_1, 2)}}.$$
Continuing this way, we can say that for some $a \in \{0,1, \ldots, a_{s_1}-1\}$ if $(a+1) {t-1 \choose s_1} > Z-Z' \geq a {t-1 \choose s_1}$ then we have
$|T|=a{t \choose s_1+1} + \left \lfloor \frac{t}{s_1+1} \right \rfloor \left(Z-Z'-a{t-1 \choose s_1} \right)$.

If $Z-Z' \geq a_{s_1} {t-1 \choose s_1}$, then we will start selecting integers from the arrays $B_{(s_2,1)}, B_{(s_2, 2)}, \ldots, B_{(s_2, a_{s_2})}$ in a similar manner. Therefore, for $a \in \{0,1, \ldots, a_{s_2}-1\}$ if $(a+1) {t-1 \choose s_2} +\alpha_1 > Z-Z' \geq a {t-1 \choose s_2} + \alpha_1$ then we have
\begin{multline*}
|T|=\underbrace{a_{s_1}{t \choose s_1+1}}_{\text{from} \ B_{(s_1,i)}, i \in[a_{s_1}]} + \underbrace{a{t \choose s_2+1}}_{\text{from} \ B_{(s_2,i)}, i \in [a]} \\
+ \underbrace{ \left \lfloor \frac{t}{s_2+1} \right \rfloor \left(Z-Z'- \alpha_1- a{t-1 \choose s_2} \right) }_{\text{from} \ B_{(s_2,a+1)}} ,
\end{multline*}
where $\alpha_1=a_{s_1} {t-1 \choose s_1}$.

Since we have $w$ elements in the set $W$, we can further continue this process until $s_w$. Therefore, in general If $(a+1) {t-1 \choose s_j} + \sum_{b=1}^{j-1} \alpha_b > Z-Z' \geq a {t-1 \choose s_j} + \sum_{b=1}^{j-1} \alpha_b$, for some $j \in \{1,2, \ldots, w\}$ and $a \in \{0,1, \ldots, a_{s_j}-1 \}$, then
\begin{multline*}
|T|= \underbrace{ \sum_{b=1}^{j-1} a_{s_b} {t \choose s_b+1}}_{\text{from} \ B_{(s_b,i)}, b \in [j-1], i \in[a_{s_b}]} + \underbrace{ a {t \choose s_j+1}}_{\text{from} \ B_{(s_j,i)}, i \in[a]} + \\
\underbrace{ \left \lfloor \frac{t}{s_j+1} \right \rfloor \left( Z-Z' - \sum_{b=1}^{j-1} \alpha_b - a {t-1 \choose s_j} \right)}_{\text{from} \ B_{(s_j,a+1)}},
\end{multline*}
where $\alpha_b = a_{s_b} {t-1 \choose s_b}$.


\end{proof}

\begin{algorithm}
	\caption{Finding a set $T$ described in Theorem \ref{thm6} }
	
	\begin{algorithmic}[1]
\For{$j \in [1: w]$}
\For{$a \in [0:a_{s_j}-1]$}
\State $x \gets \sum_{b=1}^{j-1} a_{s_b} {t-1 \choose s_b}.$
\If{$(a+1) {t-1 \choose s_j} + x > Z-Z' \geq a {t-1 \choose s_j} + x$ }
\State $T_1\gets \bigcup_{b=1}^{j-1}\bigcup_{i=1}^{a_{s_b}} S_{(s_b, i)},$ 
\State $ \qquad  \qquad \qquad\left( |T_1|=\sum_{b=1}^{j-1}a_{s_b} {t \choose s_b+1} \right),$
\State $T_2\gets \bigcup_{i=1}^{a} S_{(s_j, i)},$ 
\State $\qquad \qquad  \qquad \left( |T_2|=a {t \choose s_j+1} \right),$
\State $z\gets Z-Z'-x-a{t-1 \choose s_j},$
\State $T_3=$Function1$(s_j, z),$ 
\State $ \qquad  \qquad \qquad\left( |T_3|=z \left \lfloor \frac{t}{s_j+1} \right \rfloor \right),$
\State $T=T_1 \cup T_2 \cup T_3$.
\EndIf
\EndFor
\EndFor	
	\end{algorithmic}
	\label{algo2}
\end{algorithm}

From Corollary \ref{thm4} and Theorem \ref{thm6}, the following result follows directly.

\begin{corollary}
For $a_s\in \{0,1, \ldots, \lambda_s^t\}, s \in [t-1]$, such that $|\mathcal{R}|> \lambda_1$, HpPDA $(P, \overline{B})$ gives a $(K, K', N)$ hotplug coded caching scheme with cache memory $M$ and rate ${R}$ as follows
$$\frac{M}{N}=\frac{\lambda_1}{\sum_{s=1}^{t-1} a_s {t \choose s}} \ \text{and} \ R=\frac{\sum_{s=1}^{t-1} a_s {t \choose s+1} - |T|}{\sum_{s=1}^{t-1} a_s {t \choose s}},$$
where array $\overline{B}$ is obtained by removing all the integers in set $T$ (defined in Theorem \ref{thm6}) from the array $B$ given in \eqref{arrayB}.
\end{corollary}

\begin{example*}[Example \ref{ex6} continued]
Consider the Case 1) for which $a_1=1$, $a_2=2$ and $|\mathcal{R}|=9$. We get
$\frac{M}{N}=\frac{7}{9}$ and the rate $R=\frac{5}{9}$. 

We have $Z-Z'=2$ and $W=\{s_1, s_2\}=\{1,2\}$. The following condition is statisfied for $j=1$ and $a=1$.
\begin{equation}\label{Z-Z'}
(a+1) {t-1 \choose s_j} + \sum_{b=1}^{j-1} \alpha_b > Z-Z' \geq a {t-1 \choose s_j} + \sum_{b=1}^{j-1} \alpha_b.
\end{equation}
 Hence, from Theorem \ref{thm6}, we have a set $T=\{ (12,1), (13,1), (23,1)\}$ statisfying condition \eqref{T'} with cardinality 
\begin{align*}
|T|&= \sum_{b=1}^{j-1} a_{s_b} {t \choose s_b+1} + a {t \choose s_j+1} + \left \lfloor \frac{t}{s_j+1} \right \rfloor \\ 
&\qquad \qquad\left( Z-Z' - \sum_{b=1}^{j-1} \alpha_b - a {t-1 \choose s_j} \right)\\
&=0+1.{3\choose 2}+\left\lfloor \frac{3}{2} \right \rfloor \left( 2 - 0 - 1.{2 \choose 1} \right) \\
&=3.
\end{align*}
Therefore, the improved rate for this case is $R=\frac{5-3}{9}=\frac{2}{9}$.

Similarly consider the Case 2) for which $a_1=2$, $a_2=1$ and $|\mathcal{R}|=9$. We get
$\frac{M}{N}=\frac{7}{9}$ and the rate $R=\frac{7}{9}$. Here  $Z-Z'=3$, and the condition \eqref{Z-Z'} is statisfied for $j=1$ and $a=1$. From Theorem \ref{thm6}, we have a set $T=\{ (12,1), (13,1), (23,1), (12,2)\}$ statisfying condition \eqref{T'} with cardinality $|T|=4$ and the corresponding improved rate is $R=\frac{7-4}{9}=\frac{3}{9}$.

For the last case (Case 3)), where $a_1=2$, $a_2=2$ and $|\mathcal{R}|=12$. We get
$\frac{M}{N}=\frac{7}{12}$ and the rate $R=\frac{8}{12}$. Here  $Z-Z'=1$, and the condition \eqref{Z-Z'} is statisfied for $j=1$ and $a=0$. Again from Theorem \ref{thm6}, we have a set $T=\{ (12,1)\}$ statisfying condition \eqref{T'} with cardinality $|T|=1$ and the corresponding improved rate is $R=\frac{8-1}{12}=\frac{7}{12}$.
\end{example*}

In Example \ref{ex6}, we got two memory-rate pairs $\left(\frac{M}{N}, R\right)$ for an $(8,3, N)$ hotplug coded caching system which are $\left(\frac{7}{9}, \frac{2}{9}\right)$ and $\left(\frac{7}{12}, \frac{7}{12}\right)$. The first pair $\left(\frac{7}{9}, \frac{2}{9}\right)$ statisfies the cut-set bound given in Lemma \ref{cutset}, and hence it is optimal. The second pair $\left(\frac{7}{12}, \frac{7}{12}\right)$ does not achieve the cut-set bound, but it gives better rate than the other scheme. The rate $R_{base}=\frac{3}{5}=0.6$ is achieved for the memory $\frac{M}{N}=\frac{7}{12}$ by using the memory sharing between two memory-rate pairs $(\frac{1}{2}, \frac{4}{5})$ and $(\frac{5}{8}, \frac{1}{2})$ achieved by the baseline scheme \cite{MT2022}. Further the rate $R_{MT}=\frac{5}{8}=0.625$ is achieved for the memory $\frac{M}{N}=\frac{7}{12}$ by using the memory sharing between the memory-rate pair $(\frac{1}{3}, 1)$ achieved by the MT scheme (\textit{first new scheme} in \cite{MT2022}) and the trivial memory-rate pair $(1,0)$. Clearly, the rate $R=\frac{7}{12}=0.58$ achieved by the improved version of $t$-Scheme is better. The comparison is given in Fig. \ref{EX6} and Table \ref{tab3}. Also, the subpaketization of the $t$-Scheme is $12$, whereas the subpaketization of the baseline scheme for the memory $\frac{M}{N}=\frac{7}{12}$ is $126$.

\begin{table}[!t]
\caption{Example \ref{ex6}:  Comparison of the schemes for a $(8,3,N)$ hotplug coded caching system. \label{tab3}}
\centering
\begin{tabular}{| m{1cm} | m{2cm}| m{1cm} | m{1cm} | m{1cm}| } 
 \hline
 & Parameters & Improved $t$-Scheme & Baseline scheme \cite{MT2022} & MT scheme  \cite{MT2022} \\ 
 \hline
\multirow{2}{4em}{$\frac{M}{N}=\frac{7}{9}$} & Subpacketization & $9$ & $36$ & $4$  \\ 
 & Rate & $0.222$ & $0.25$ & $0.333$ \\
\hline
\multirow{2}{4em}{$\frac{M}{N}=\frac{7}{12}$} & Subpacketization & $12$ & $126$ & $4$  \\ 
 & Rate & $0.58$ & $0.6$ & $0.625$ \\
\hline
\end{tabular}
\end{table}

 \begin{figure*}
  \includegraphics[width=0.83\linewidth]{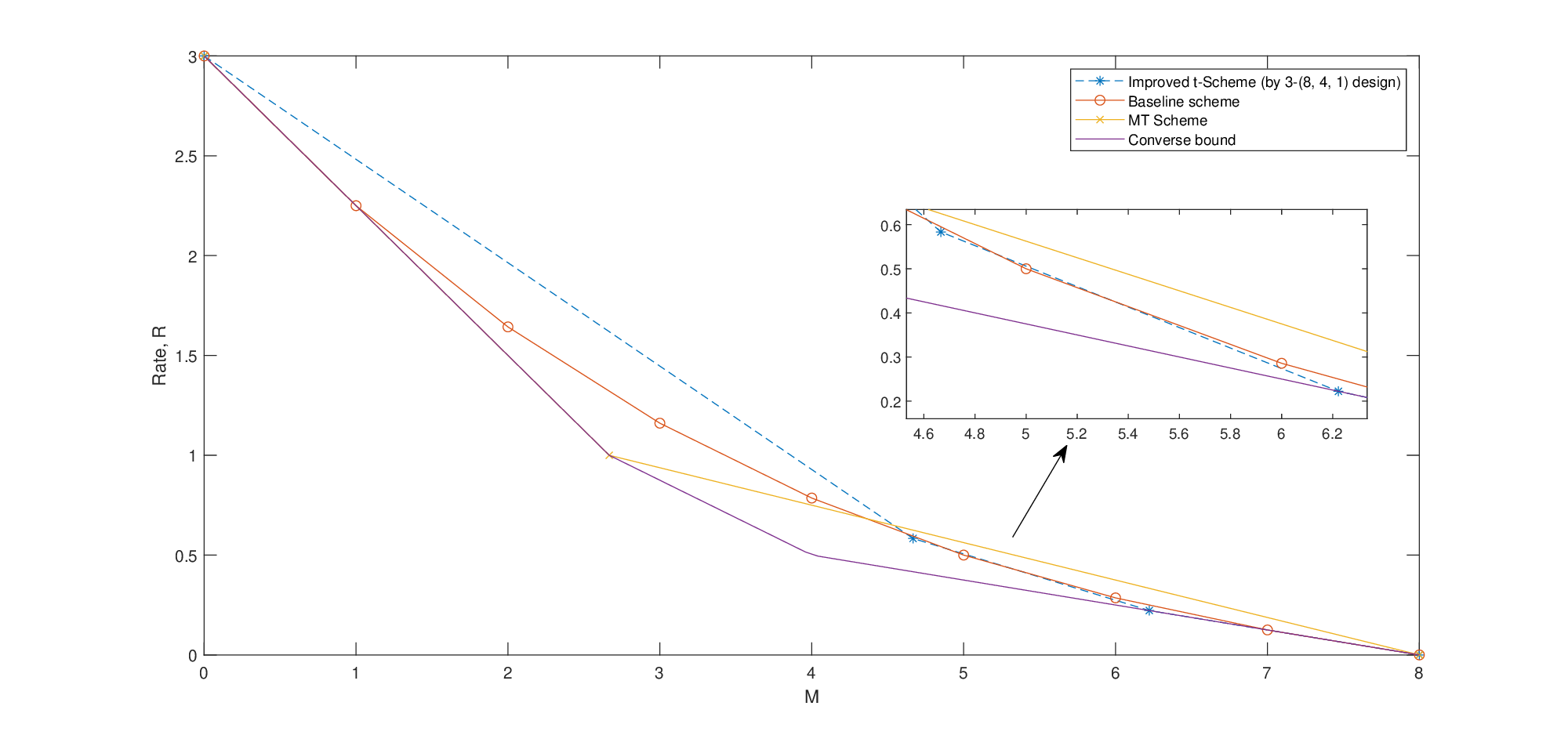}
  \caption{Example 6: (8, 3, 8) hotplug coded caching system.}
  \label{EX6}
\end{figure*}

\section{Numerical evaluations}\label{num}
In this section, we consider several other examples of $t$-designs and show that we are getting better rate using improved $t$-Scheme than the existing schemes in some higher memory ranges. Also, we get some memory-rate points which achieves the cut-set bound. 

There exists a $5$-$(12,6,1)$ design and a construction of this design is given in \cite{S2004}. By using improved $t$-Scheme for this design we get a rate for $(12,5, 12)$ hotplug coded caching system shown in Fig. \ref{EX7}. Further, as described in Remark \ref{rem6} that a $t$-$(v, k, \lambda)$ design $(X, \mathcal{A})$ is also an $s$-$(v, k, \lambda_s)$ design, for all $1 \leq s \leq t$. Therefore, there exists a $4$-$(12, 6, 4)$ design using which we get a $(12, 4, N)$ hotplug coded caching scheme with $4$ number of active users using improved $t$-Scheme, and Fig. \ref{EX8} shows the memory-rate tradeoff corresponding to it.
Similarly, a $3$-$(12, 6, 12)$ design obtained from $5$-$(12,6,1)$ design will correspond to a $(12, 3, N)$ hotplug coded caching scheme for which memory-rate tradeoff is given in Fig. \ref{EX9} with $N=12$. It is clear from Fig. \ref{EX9} that the improved $t$-Scheme achives the cut-set bound from the memory point $M=9.432$, while the baseline scheme achives it from $M=11$. 

 \begin{figure*}
  \includegraphics[width=0.83\linewidth]{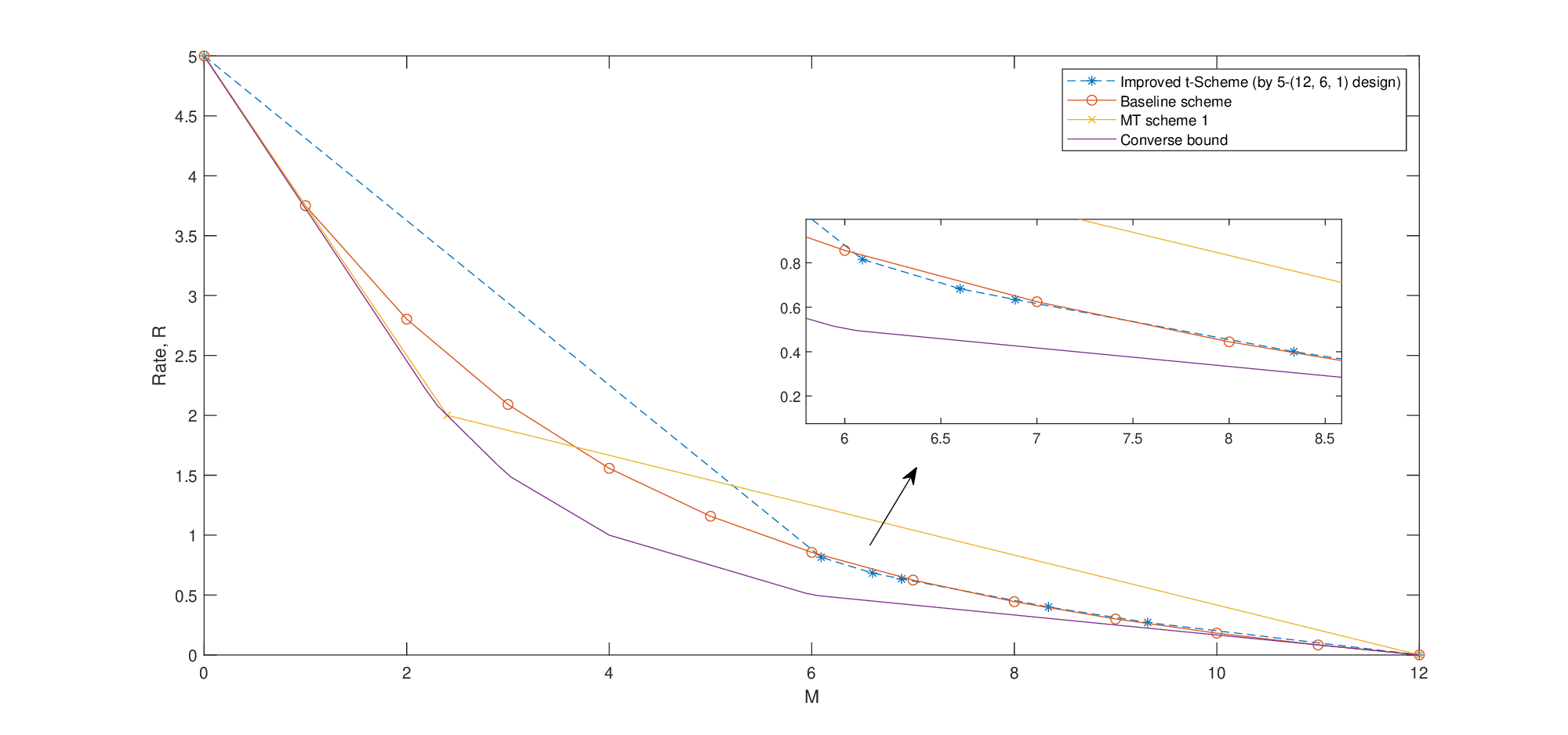}
  \caption{(12, 5, 12) hotplug coded caching system.}
  \label{EX7}
\end{figure*}
 \begin{figure*}
  \includegraphics[width=0.83\linewidth]{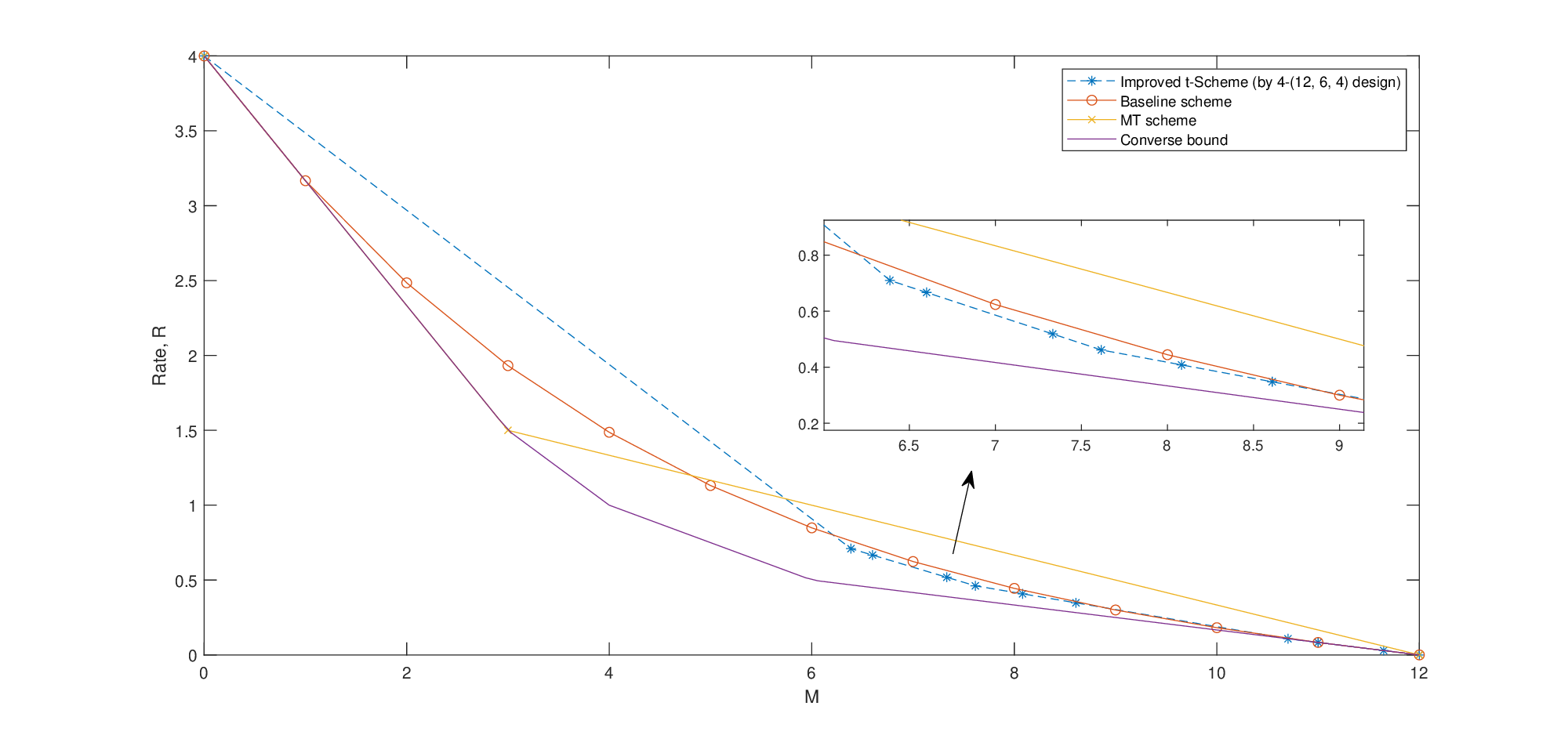}
  \caption{(12, 4, 12) hotplug coded caching system.}
  \label{EX8}
\end{figure*}
\begin{figure*}
  \includegraphics[width=0.83\linewidth]{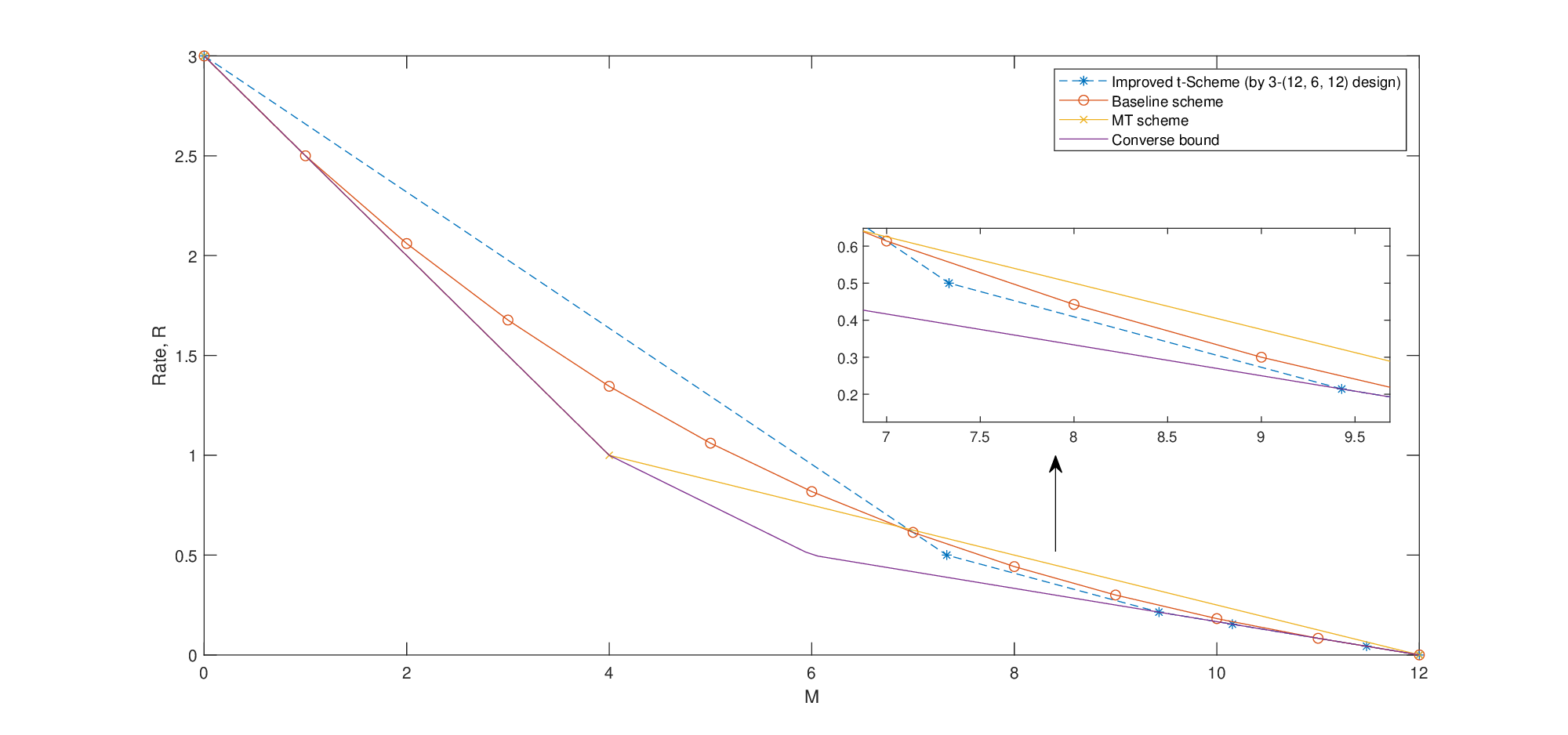}
  \caption{(12, 3, 12) hotplug coded caching system.}
  \label{EX9}
\end{figure*}

There exist another $3$-$(12,4,3)$ design with $v=12$ points and $t=3$, which also corresponds to a $(12, 3, N)$ hotplug coded caching scheme. In Fig. \ref{EX10}, we compare the memory-rate tradeoff of the schemes obtained by $3$-$(12,6,12)$ design and  $3$-$(12,4,3)$ design for a $(12, 3, 12)$ hotplug coded caching system.

\begin{figure*}
  \includegraphics[width=0.83\linewidth]{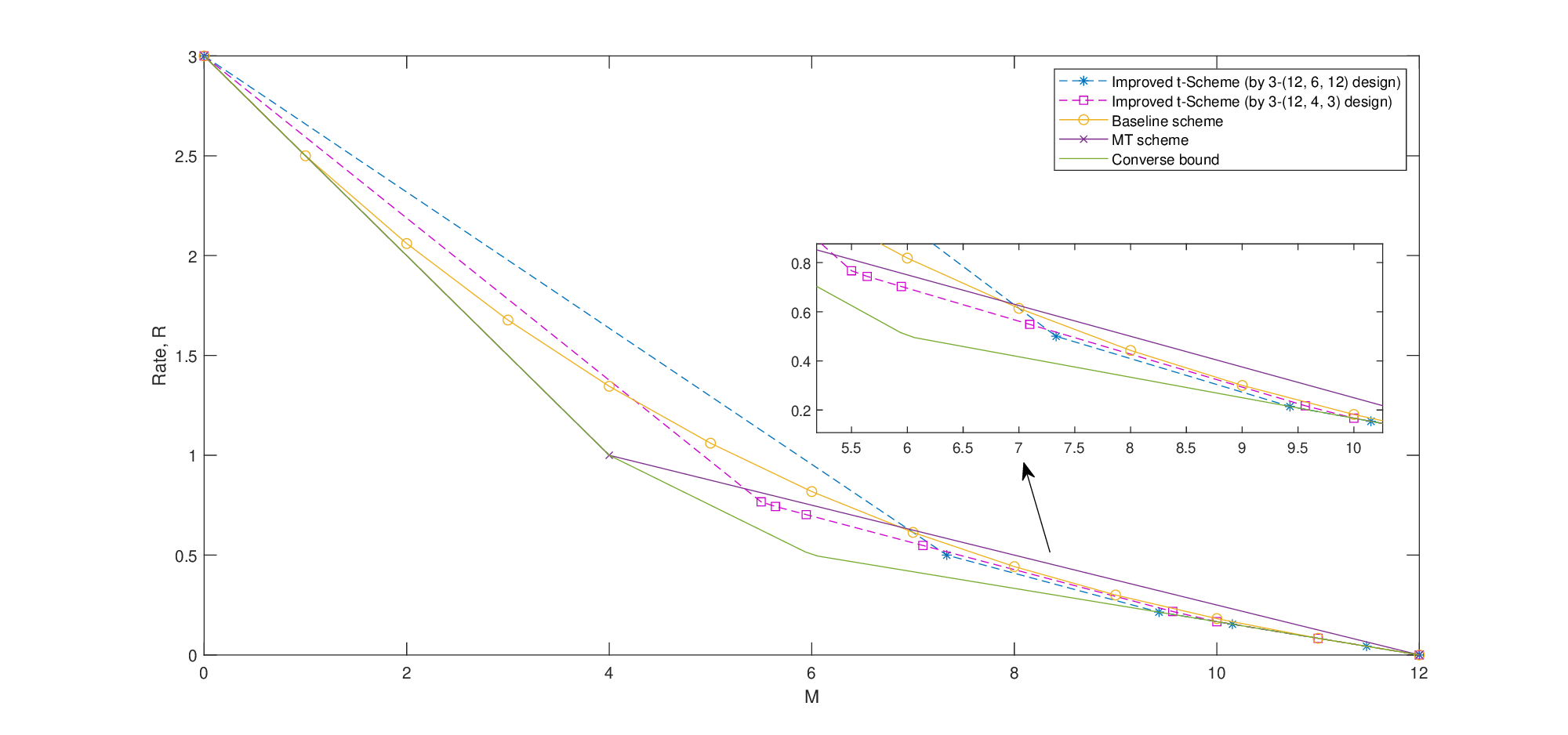}
  \caption{(12, 3, 12) hotplug coded caching system.}
  \label{EX10}
\end{figure*}

In a similar way, we obtain $2$-$(12,6,30)$ design from $5$-$(12,6,1)$ design and a $2$-$(12,4,15)$ design from $3$-$(12,4,3)$ design. The both designs correspond to a $(12, 2, N)$ hotplug coded caching scheme and the comparison of the memory-rate tradeoff of these  schemes is given in Fig. \ref{EX11} for $N=12$. It can be seen that both schemes achieves the cut-set bound in the higher memory ranges. However, the MT scheme \cite{MT2022} achieves the converse bound completely for a $(12, 2, 12)$ hotplug coded caching system.

\begin{figure}
  \includegraphics[width=\linewidth]{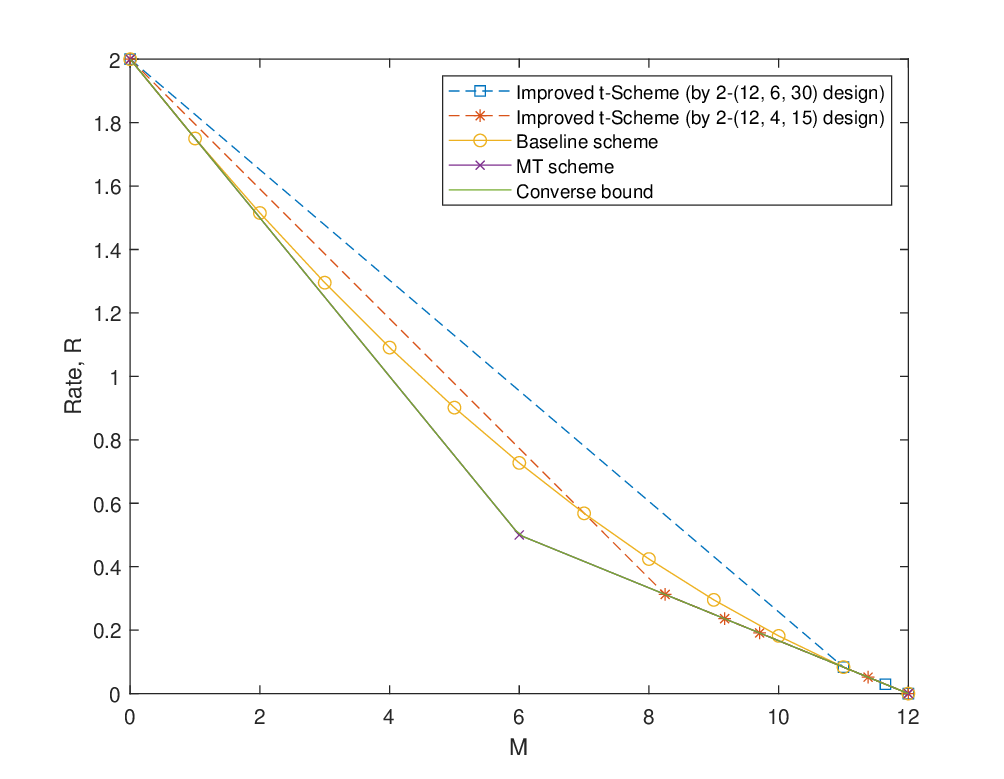}
  \caption{(12, 2, 12) hotplug coded caching system.}
  \label{EX11}
\end{figure}

\section{Optimality} \label{OPTIMAL}
In this section, we focus on the hotplug coded caching schemes obtained using the improved $t$-Scheme from the following two clasess of $t$-designs. 
\begin{itemize}
\item For all even integers $v \geq 6$, there exists a $3$-$(v,4,3)$ design. 
\item For all prime powers $q$, there exists a $3$-$(q^2+1, q+1, 1)$ design.
\end{itemize}
Using these clasess of designs, we get some optimal memory-rate points achieving cut-set bound for $(v, 3, N)$ hotplug coded caching system.
 
\begin{theorem}\label{opti}
For a $(K, 3, N)$ hotplug coded caching system, where $K$ is an even integer and $K \geq 6$, we get the optimal memory-rate points 
$$\left(\frac{M}{N}, R \right)=\left(\frac{{K-1 \choose 2}}{3(a_1+a_2)}, \frac{3(a_1+a_2)-{K-1 \choose 2}}{3(a_1+a_2)}\right),$$
where $0\leq a_1 \leq {K-4 \choose 2}$, $0 \leq a_2 \leq \frac{3}{2}(K-4)$ and 
${K-1 \choose 2}=a+3a_1+2a_2$ for some $0 \leq a < a_2$. 
\end{theorem}

\begin{example} \label{ex7}
Consider a $(K, 3, N)$ hotplug coded caching system, where $K=12$. We have $0\leq a_1 \leq 28$ and $0 \leq a_2 \leq 12$. For $a_1=10$ and $a_2=12$, the condition ${K-1 \choose 2}=a+3a_1+2a_2$ is satisfied for $a=1$. Therefore, we have the following optimal memory rate point
$$\left(\frac{M}{N}, R \right)=\left(\frac{5}{6}, \frac{1}{6}\right),$$
which satisfy the cut-set bound as $R=1-\frac{M}{N}$. Similarly, we get the following points for different choices of $a_1$ and $a_2$.
\begin{itemize}
\item For $a_1=11$ and $a_2=11$, we get $\left(\frac{M}{N}, R \right)=\left(\frac{55}{66}, \frac{11}{66}\right)$.\\
\item For $a_1=9$ and $a_2=12$, we get $\left(\frac{M}{N}, R \right)=\left(\frac{55}{63}, \frac{8}{63}\right)$.\\
\item For $a_1=8$ and $a_2=12$, we get $\left(\frac{M}{N}, R \right)=\left(\frac{55}{60}, \frac{5}{60}\right)$.\\
\item For $a_1=7$ and $a_2=12$, we get $\left(\frac{M}{N}, R \right)=\left(\frac{55}{57}, \frac{2}{57}\right)$.
\end{itemize}
The memory-rate tradeoff for this example is given in Fig. \ref{OS12}.
\end{example}

\begin{figure*}
 \includegraphics[width=0.83\linewidth]{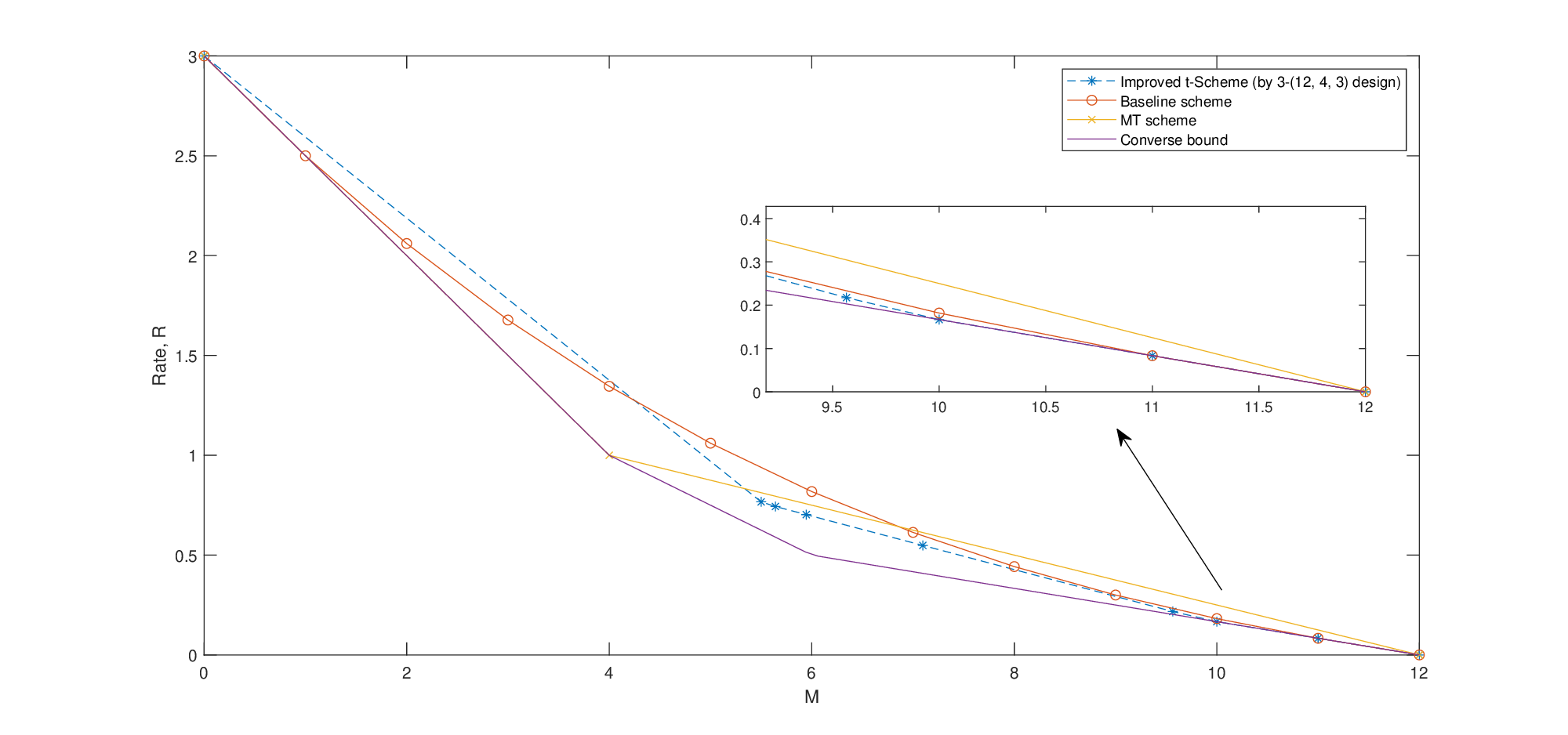}
  \caption{Example \ref{ex7}: (12, 3, 12) hotplug coded caching system.}
  \label{OS12}
\end{figure*}

\begin{theorem}\label{optiprime}
For a $(q^2+1, 3, N)$ hotplug coded caching system, where $q$ is a prime power and $q \geq 3$, we get the optimal memory-rate points 
$$\left(\frac{M}{N}, R \right)=\left(\frac{q(q+1)}{3(a_1+a_2)}, \frac{3(a_1+a_2)-q(q+1)}{3(a_1+a_2)}\right),$$
where $0\leq a_1 \leq q^2-q-1$, $0 \leq a_2 \leq q$ and $q(q+1)=a+3a_1+2a_2$ for some $0 \leq a < a_2$. 
\end{theorem}

\begin{example}\label{ex8}
Consider a $(K, 3, N)$ hotplug coded caching system, where $K=q^2+1$ and $q=4$, i.e., a $(17, 3, N)$ hotplug coded caching system. We have $0\leq a_1 \leq 11$ and $0 \leq a_2 \leq 4$. For $a_1=3$ and $a_2=4$, the condition $q(q+1)=a+3a_1+2a_2$ is satisfied for $a=3$. Therefore, we have the following optimal memory rate point
$$\left(\frac{M}{N}, R \right)=\left(\frac{20}{21}, \frac{1}{21}\right),$$
which satisfy the cut-set bound as $R=1-\frac{M}{N}$. Similarly, we get the following points for different choices of $a_1$ and $a_2$.
\begin{itemize}
\item For $a_1=5$ and $a_2=2$, we get $\left(\frac{M}{N}, R \right)=\left(\frac{20}{21}, \frac{1}{21}\right)$.\\
\item For $a_1=4$ and $a_2=4$, we get $\left(\frac{M}{N}, R \right)=\left(\frac{20}{24}, \frac{4}{24}\right)$.
\end{itemize}
The memory-rate tradeoff for this example is given in Fig. \ref{OSP17}.
\end{example}

\begin{figure*}
 \includegraphics[width=0.83\linewidth]{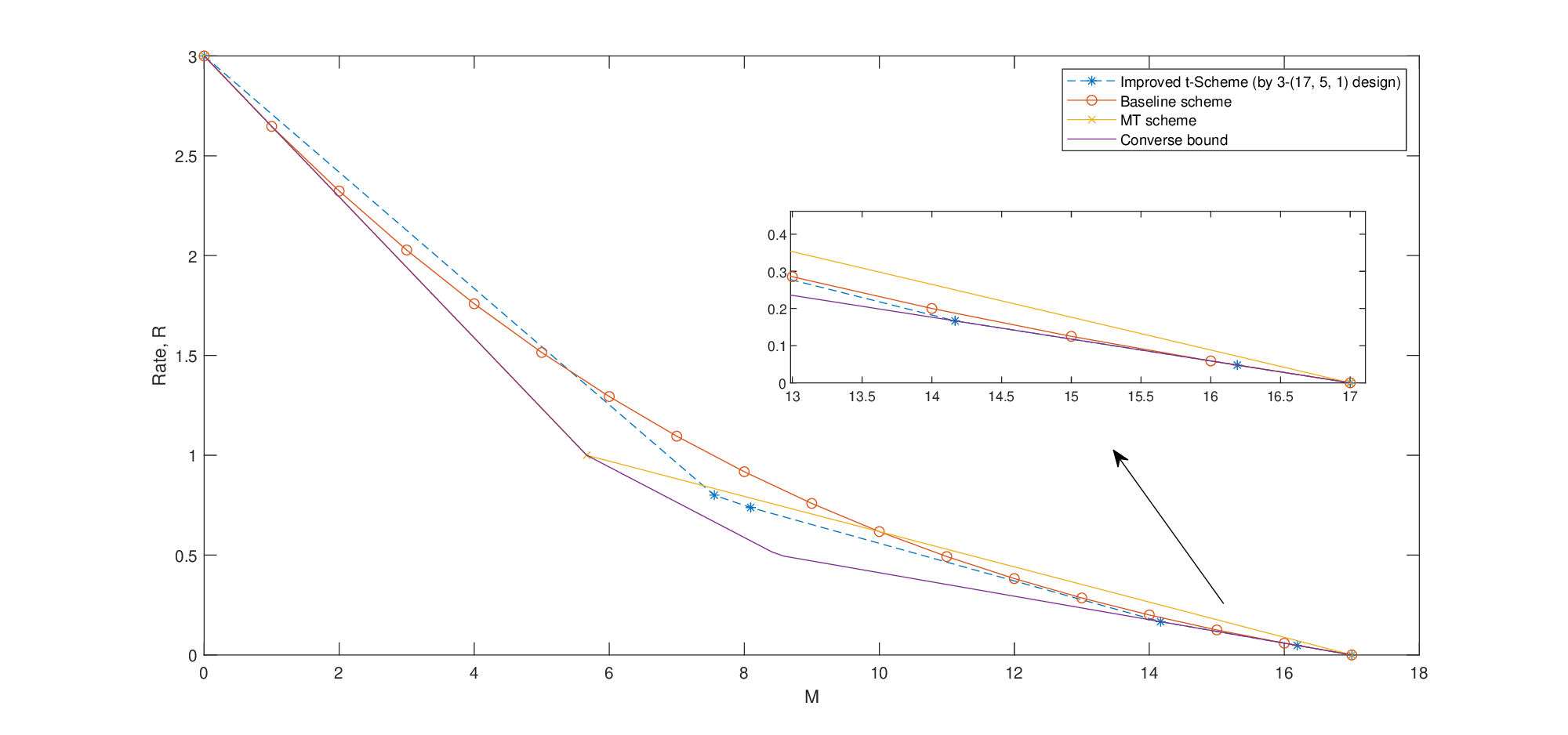}
  \caption{Example \ref{ex8}: (17, 3, 17) hotplug coded caching system.}
  \label{OSP17}
\end{figure*}

In Fig. \ref{OS&OSP26}, we give a compared memory-rate tradeoffs for a $(26, 3, 20)$ hotplug coded caching system, one obtained from a $3$-$(v,4,3)$ design with $v=26$ and other obtained from a $3$-$(q^2+1, q+1, 1)$ design with $q=5$. The Fig. \ref{OS&OSP26} indicates that each design outperforms the other within certain memory ranges.

\begin{figure*}
 \includegraphics[width=0.83\linewidth]{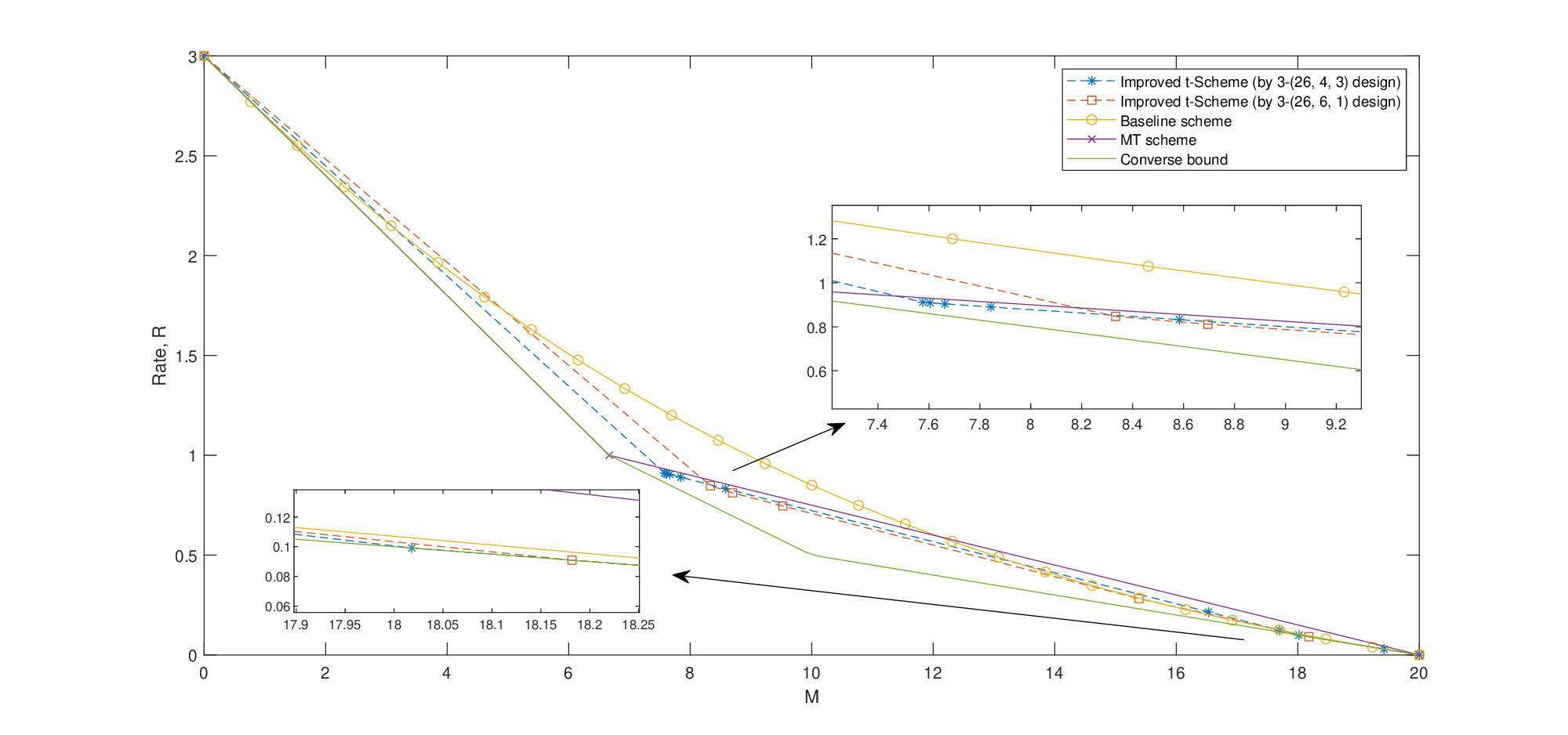}
  \caption{(26, 3, 20) hotplug coded caching system.}
  \label{OS&OSP26}
\end{figure*}

\subsection{\textbf{$(K, 3, N)$ hotplug coded caching scheme from a $3$-$(v, k, \lambda)$ design}}
In this subsection, we find the memory-rate points for a $(K, 3, N)$ hotplug coded caching scheme using a $3$-$(v, k, \lambda)$ design. Using $t$-Scheme for this design, we have
\begin{align*}
Z&=\lambda_1 \\
|\mathcal{R}|&=3a_1+3a_2, \\
Z'&=a_1+2a_2, \\
S&=3a_1+a_2, \\
\frac{M}{N}=\frac{\lambda_1}{3a_1+3a_2} \ &\text{and} \ R=\frac{3a_1+a_2}{3a_1+3a_2},
\end{align*}
where $a_1 \in \{0,1, \ldots, \lambda_1^3\}$ and $a_2 \in \{0,1, \ldots, \lambda_2^3\}$.
Since we have $Z-Z'=\lambda_1-(a_1+2a_2)$, using Theorem \ref{thm6}, we further improve the rate. Considering the case when $a_1$ and $a_2$ both are non-zero, we have $W=\{1,2\}=\{s_1, s_2\}$ and $\alpha_1=a_1 {2 \choose 1}=2a_1$. Now for $j=2$, the condition 
$$(a+1) {t-1 \choose s_j} + \sum_{b=1}^{j-1} \alpha_b > Z-Z' \geq a {t-1 \choose s_j} + \sum_{b=1}^{j-1} \alpha_b$$
given in Theorem \ref{thm6} is satisfied if
$$(a+1) {2 \choose 2} + 2 a_1  >  \lambda_1-(a_1+2a_2) \geq a {2 \choose 2} +  2 a_1 $$
 for some $a \in \{0,1, \ldots, a_2 \}$. That means
$$(a+1) + 2a_1  >  \lambda_1-(a_1+2a_2) \geq a + 2 a_1 $$
 or 
$(a+1) + 3a_1 +2 a_2  >  \lambda_1 \geq a + 3a_1 +2 a_2,$
which is equivalent to 
\begin{equation}
\lambda_1=a + 3a_1 +2 a_2.
\end{equation}
Further, for $j=2$ and $a \in \{0,1, \ldots, a_2 \}$, we have 
\begin{align*}
|T|&= a_1{3 \choose 2} + a {3 \choose 3} + \left \lfloor \frac{3}{3} \right \rfloor \left( Z-Z' - \alpha_1 - a {2 \choose 2} \right)\\
&=3a_1+a+(\lambda_1-(a_1+2a_2)-2a_1 -a)\\
&=\lambda_1 -2a_2.
\end{align*}
Therefore, the improved rate is 
$$R=\frac{S-|T|}{|\mathcal{R}|}=\frac{3a_1+3a_2-\lambda_1}{3a_1+3a_2},$$
and we get the following memory-rate point
\begin{equation}\label{optipoints}
\left(\frac{M}{N}, R \right)=\left(\frac{\lambda_1}{3(a_1+a_2)}, \frac{3(a_1+a_2)-\lambda_1}{3(a_1+a_2)}\right),
\end{equation}
where $0\leq a_1 \leq \lambda_1^3$, $0 \leq a_2 \leq \lambda_2^3$ and $\lambda_1=a+3a_1+2a_2$ for some $0 \leq a < a_2$. 

Using the method given above, we get the proof of Theorem \ref{opti} and Theorem \ref{optiprime} using a $3$-$(v,4,3)$ design, where $v \geq 6$ is an even integer, and 
$3$-$(q^2+1, q+1, 1)$ design, where $q$ is a prime power, respectively.

\begin{proof}[Proof of Theorem \ref{opti}]
Consider $3$-$(v,4,3)$ design, where $v \geq 6$ is an even integer. We have 
\begin{align*}
\lambda_1 &= \frac{\lambda {v-1 \choose t-1}}{{k-1 \choose t-1}}= \frac{3 {v-1 \choose 2}}{{3 \choose 2}} = {v-1 \choose 2}, \\
\lambda_1^3 & = \frac{\lambda {v-t \choose k-1}}{{v-t \choose k-t}}=\frac{3 {v-3 \choose 3}}{{v-3 \choose 1}} = {v-4 \choose 2}, \\
  \lambda_2^3 & = \frac{\lambda {v-t \choose k-2}}{{v-t \choose k-t}} =\frac{3 {v-3 \choose 2}}{{v-3 \choose 1}} = \frac{3}{2}(v-4). 
\end{align*} 
Using these value in equation \eqref{optipoints}, we get the optimal memory-rate points for a $(K, 3, N)$ hotplug coded caching system, where $K=v$ is an even integer and $K \geq 6$.
\end{proof}

\begin{proof}[Proof of Theorem \ref{optiprime}]
Consider $3$-$(q^2+1, q+1, 1)$ design, where $q \geq 3$ is a prime power. We have 
\begin{align*}
\lambda_1 &= \frac{\lambda {v-1 \choose t-1}}{{k-1 \choose t-1}}=  \frac{{q^2 \choose 2}}{{q \choose 2}} = q(q+1), \\
\lambda_1^3 & = \frac{\lambda {v-t \choose k-1}}{{v-t \choose k-t}} = \frac{ {q^2-2 \choose q}}{{q^2-2 \choose q-2}} = q^2-q-1, \\
  \lambda_2^3 & = \frac{\lambda {v-t \choose k-2}}{{v-t \choose k-t}} = \frac{ {q^2-2 \choose q-1}}{{q^2-2 \choose q-2}}  = q. 
\end{align*} 
Using these value in equation \eqref{optipoints}, we get the optimal memory-rate points for a $(q^2+1, 3, N)$ hotplug coded caching system, where  $q \geq 3$ is a prime power.
\end{proof}

\subsection{Comparison with the baseline scheme \cite{MT2022}}
The baseline scheme is a classical YMA scheme \cite{YMA2018} with restricted set of demand vectors $d \in [\min (N, K')]$. For a $(K, K', N)$ hotplug coded caching system, the lower convex envelope of the following memory-rate points are achievable by the baseline scheme.
$$\left(\frac{M_t}{N}, R_t \right) = \left( \frac{{K-1 \choose t-1}}{{K \choose t}}, \frac{{K \choose t+1} - {K-r' \choose t+1}}{{K \choose t}} \right), $$
 for all $t \in \{0,1, \ldots, K\}$, where $r'=\min(N,K')$. This scheme achieves the cut-set bound for $t=K-1, K$, i.e., the baseline scheme start achieving the cut-set bound from the memory $M=N\left(\frac{K-1}{K} \right)$.

Now we show that for a $(K, 3, N)$ hotplug coded caching system, our improved $t$-Scheme start achieving the cut-set bound from the memory less than $N\left(\frac{K-1}{K}\right)$ using Theorem \ref{opti} and Theorem \ref{optiprime}.

For a $3$-$(v,k, \lambda)$ design, $\lambda_1$ is fixed. Therefore, from equation \ref{optipoints}, we get the optimal rate for the memory
$$\frac{M_{min}}{N}=\min\left(\frac{\lambda_1}{|\mathcal{R}|}\right)=\frac{\lambda_1}{\max(|\mathcal{R}|)},$$
where $0\leq a_1 \leq \lambda_1^3$, $0 \leq a_2 \leq \lambda_2^3$ and $\lambda_1=a+3a_1+2a_2$ for some $0 \leq a < a_2$. Since $|\mathcal{R}|=3(a_1+a_2)=\lambda_1+a_2-a$, we have $\max(|\mathcal{R}|)=\lambda_1+\max(a_2)-a=\lambda_1+\lambda_2^3-a$. After choosing maximum value of $a_2$ which is $\lambda_2^3$, we get $a_1=\frac{\lambda_1-2\lambda_2^3-a}{3}$ for some $a, 0 \leq a < a_2$. Here $a_1$ will be an integer for $0 \leq a < 3$. Therefore, we have 
$$\frac{M_{min}}{N}=\frac{\lambda_1}{\lambda_1+\lambda_2^3-a},$$
where $0 \leq a < 3$ such that $a_1=\frac{\lambda_1-2\lambda_2^3-a}{3}$ is an integer.
Now we prove that the memory point $\frac{M_{min}}{N}$ is less than $\frac{K-1}{K}$ for the following cases.
\begin{enumerate}
\item  For a $(K, 3, N)$ hotplug coded caching system using $3$-$(v,4,3)$ design, where $v$ is an even integer and $v \geq 6$, we have
$\lambda_1 = {v-1 \choose 2}, \lambda_1^3 = {v-4 \choose 2}$ and $\lambda_2^3 = \frac{3}{2}(v-4).$  Therefore, we have 
$$\frac{M_{min}}{N}=\frac{{v-1 \choose 2}}{{v-1 \choose 2}+\frac{3}{2}(v-4)-a},$$
where $0\leq a \leq 2$. Since $v=K$ and $a \leq 2$,
\begin{align*}
\frac{M_{min}}{N} &\leq \frac{{K-1 \choose 2}}{{K-1 \choose 2}+\frac{3}{2}(K-4)-2} \\
&=\frac{(K-1)(K-2)}{(K-1)(K-2)+3(K-4)-4}\\
&\leq\frac{(K-1)(K-2)}{(K-1)(K-2)+3(K-2)} \\
&=\frac{K-1}{K+2} <\frac{K-1}{K}.
\end{align*}
\item  For a $(K, 3, N)$ hotplug coded caching system using $3$-$(q^2+1,q+1,1)$ design, where $q$ is a prime power, $q \geq 3$ and $K=q^2+1$, we have
$\lambda_1 = q(q+1), \lambda_1^3 = q^2-q-1$ and $\lambda_2^3 = q.$  Therefore, we have 
$$\frac{M_{min}}{N}=\frac{q(q+1)}{q(q+1)+q-a},$$
where $0\leq a \leq 2$. Since $a \leq 2$,
\begin{align*}
\frac{M_{min}}{N} &\leq \frac{q(q+1)}{q(q+1)+q-2}\\
&=\frac{q(q+1)}{(q+1)^2-3} < \frac{q(q+1)}{(q+1)^2} \\
&=\frac{q}{(q+1)}= \frac{q^2}{q(q+1)}\\
&<  \frac{q^2}{q^2+1}=  \frac{K-1}{K}.
\end{align*}
\end{enumerate}

\section{Conclusion}\label{conclu}
In this paper, we considered a practical scenario of a coded caching model in which some users are inactive at the time of delivery, known as a hotplug coded caching system. We introduced a structure called HpPDA that corresponds to a hotplug coded caching scheme, and then we described the placement and delivery phase of that scheme in Algorithm \ref{algo1}. We provided a method to further improve the rate of the proposed hotplug coded caching scheme using HpPDA. Further, we presented a construction of HpPDAs using $t$-designs and then presented a way to improve the rate. We compare the rate of the scheme obtained from HpPDAs using $t$-designs with the existing schemes. For a hotplug coded caching system with three active users, we proved that the cut-set bound is achieved for some higher memory range. Following are the possible directions for further research.
\begin{enumerate}
\item Two distinct classes of HpPDAs have been introduced in this paper. The first class consists of MAN HpPDAs, and the second class comes from the $t$-designs. However, there exist other HpPDAs which does not fall into these classes, such as Example \ref{ex3}, which can not be obtained by any of these methods. Therefore, exploring further classes of HpPDAs is an interesting problem.
\item Exploring different combinatorial designs from existing literature and using them to create a new class of HpPDAs is an exciting direction.
\item The problem of high subpacketization still remains with the MAN HpPDAs. However, the subpacketization is lower in HpPDAs obtained from $t$-designs compared to the MAN HpPDAs. It will be interesting to find new classes of HpPDAs with lower subpacketization.
\end{enumerate}

\section*{Acknowledgments}
This work was supported partly by the Science and Engineering Research Board (SERB) of the Department of Science and Technology (DST), Government of India, through J.C. Bose National Fellowship to Prof. B. Sundar Rajan and by the C V Raman postdoctoral fellowship awarded to Charul Rajput.

%
%

\end{document}